\def\llncs{1}
\def\fullpage{1}
\def\draft{1}
\def\anonymous{0}

\def\submission{0}
\ifnum\submission=1
    \def\llncs{1}
    \def\anonymous{1}
    \def\draft{0}
\fi

\ifnum\llncs=1
    \documentclass{llncs}
\else
    \documentclass[letterpaper,hmargin=1.05in,vmargin=1.05in,11pt]{article}
		\ifnum\fullpage=1
    \usepackage{fullpage}
    \fi
\fi

\usepackage{amsmath, amsfonts, amssymb, mathtools,amscd}

\usepackage{amsthm}

\usepackage{lmodern}
\usepackage[T1]{fontenc}
\usepackage{multirow}
\usepackage{makecell}

\usepackage[utf8]{inputenc}
\usepackage{graphicx}
\usepackage{subcaption}
\usepackage{authblk}
\usepackage[usenames]{color}
\usepackage{xcolor,colortbl} % Leave out in case the usepackage cannot be found. Not critical
\definecolor{darkblue}{rgb}{0,0,0.6}
\definecolor{darkgreen}{rgb}{0,0.5,0}
\definecolor{maroon}{rgb}{0.5,0.1,0.1}
\definecolor{dpurple}{rgb}{0.2,0,0.65}

\usepackage{physics}
\usepackage{comment}

\usepackage{algorithm}
\usepackage[noend]{algpseudocode}

\usepackage[colorlinks=true,linkcolor=magenta,citecolor=blue]{hyperref}
\usepackage[nameinlink]{cleveref}

\usepackage{xargs}                      % Use more than one optional parameter in a new commands
\usepackage{xcolor}  % Coloured text etc.
\usepackage[colorinlistoftodos,prependcaption,textsize=tiny]{todonotes}

\usetikzlibrary{calc}
\usepackage{pgfplots}
\pgfplotsset{compat=1.18}

\newcommand{\N}{\mathbb{N}}

\newcommand{\R}{\mathbb{R}}

\newcommand{\rQ}{\mathbf Q}
\newcommand{\rW}{\mathbf W}
\newcommand{\rT}{\mathbf T}
\newcommand{\rX}{\mathbf X}

\newcommand{\cC}{\mathcal{C}}
\newcommand{\cU}{\mathcal{U}}
\newcommand{\cW}{\mathcal{W}}

\newcommand{\Sd}{\mathcal{S}^{d-1}}
\renewcommand{\H}{\mathcal{H}}

\newcommand{\stdBO}{\textsf{StdBO}}
\newcommand{\compBO}{\textsf{RBO}}

\newcommand{\poly}{{\rm poly}}
\newcommand{\inner}[1]{\left\langle{#1}\right\rangle}

\newtheorem{assumption}{Assumption}
\newtheorem{costmodel}{Cost Model}

\newcommand{\vecb}{{\vec{b}}}
\newcommand{\vecc}{{\vec{c}}}
\newcommand{\vecv}{{\vec{v}}}
\newcommand{\vecw}{{\vec{w}}}
\newcommand{\vecx}{{\vec{x}}}
\newcommand{\vect}{{\vec{t}}}

\newcommand{\vecy}{{\vec{y}}}
\newcommand{\vecz}{{\vec{z}}}
\newcommand{\vecu}{{\vec{u}}}

\DeclareMathOperator*{\argmin}{arg\,min}

\ifnum\draft=1
    \newcommand{\minki}[1]{$\ll$\textsf{\color{darkgreen} Minki: { #1}}$\gg$}
    \newcommand{\yixin}[1]{$\ll$\textsf{\color{blue} Yixin: { #1}}$\gg$}
    \newcommand{\jh}[1]{$\ll$\textsf{\color{purple} JeongHoon: { #1}}$\gg$}
    \newcommand{\bg}[1]{$\ll$\textsf{\color{orange} BeomGeun: { #1}}$\gg$}
    \newcommandx{\todominki}[2][1=]{\todo[linecolor=red,backgroundcolor=red!25,bordercolor=red,#1]{#2}}
    \newcommandx{\todojh}[2][1=]{\todo[linecolor=blue,backgroundcolor=blue!25,bordercolor=blue,#1]{#2}}
    \newcommandx{\todobg}[2][1=]{\todo[linecolor=orange,backgroundcolor=orange!25,bordercolor=orange,#1]{#2}}
\else
    \newcommand{\minki}[1]{}
    \newcommand{\yixin}[1]{}
    \newcommand{\jh}[1]{}
    \newcommand{\bg}[1]{}
    \newcommandx{\todominki}[2][1=]{}
    \newcommandx{\todojh}[2][1=]{}
\fi

\newcommand{\sbt}{\,\begin{picture}(-1,1)(-1,-3)\circle*{3}\end{picture}\ }

\title{Does quantum lattice sieving require quantum RAM?}
\ifnum\anonymous=1
    \ifnum\llncs=1
        \author{\empty}\institute{\empty}
    \else
        \author{}
    \fi
\else
    \ifnum\llncs=1
        \author{Beomgeun Cho\inst{1} \and Minki Hhan$^\star$\inst{2}\and Taehyun Kim\inst{1} \and Jeonghoon Lee\inst{1} \and Yixin Shen\inst{3}}
        \institute{Seoul National University, Seoul, Republic of Korea, \email{\{c11sh0117,taehyun,dalsan2113\}@snu.ac.kr} \and
                   The University of Texas at Austin, Texas, USA, \email{minki.hhan@austin.utexas.edu} \and Univ Rennes, Inria, CNRS, IRISA, Rennes, France 
                   \email{yixin.shen@inria.fr}}
    \else
        \author[1]{Minki Hhan}
        \affil[1]{{\small The University of Texas at Austin, Texas, USA}\authorcr{\small minki.hhan@austin.utexas.edu}}
    \fi
\fi

\begin{document}
\maketitle
% use optional argument because the \LaTeX command breaks the PDF keywords

\begin{abstract}
    In this paper, we study the requirement for \emph{quantum} random access memory (QRAM) in quantum lattice sieving, a fundamental algorithm for lattice-based cryptanalysis.

    First, we obtain a lower bound on the cost of quantum lattice sieving with a bounded size QRAM. We do so
    in a new query model encompassing a wide range of lattice sieving algorithms similar to those in the classical sieving lower bound by Kirshanova and Laarhoven [CRYPTO 21].
    This implies that, under reasonable assumptions, quantum speedups in lattice sieving require the use of QRAM.
    In particular, no quantum speedup is possible without QRAM.

    Second, we investigate the trade-off between the size of QRAM and the quantum speedup.
    We obtain a new interpolation between classical and quantum lattice sieving.
    Moreover, we show that further improvements require a novel way to use the QRAM by proving the optimality of some subroutines. An important caveat is that this trade-off requires a strong assumption on the efficient replacement of QRAM data, indicating that even speedups with a small QRAM are already challenging.

    Finally, we provide a circuit for quantum lattice sieving without using QRAM. Our circuit has a better depth
    complexity than the best classical algorithms but requires an exponential amount of qubits.
    To the best of our knowledge, this is the first quantum speedup for lattice sieving without QRAM in the standard quantum circuit model. We explain why this circuit does not contradict our lower bound, which considers the query complexity.
\end{abstract}


% \newpage

\setcounter{footnote}{1}
\section{Introduction}
\renewcommand*{\thefootnote}{\fnsymbol{footnote}}
\footnotetext{Most of this work was done while Minki Hhan was in KIAS, Korea.}
\renewcommand*{\thefootnote}{\arabic{footnote}}
\setcounter{footnote}{0}
% \minki{What do you think about title? The current version focuses on the algorithmic side, but I think the lower bound is also interesting. I think the I guess sth like ``Does the quantum lattice sieving require quantum RAM?'' or ``The need for QRAM in quantum lattice sieving'' is succinct and funny (I am inspired by the title ``Is there an oblivious RAM lower bound?'' and ``the need for structure in quantum speedups'')}
% \yixin{I like both your suggestions better than the current title.}
% \minki{I changed to ``Does the quantum lattice sieving require quantum RAM?'', because we also have a sieving algorithm without QRAM (but with exponentially large quantum memory, though)}
% \minki{As the title is the question, I feel we should leave the open problem for answering the question with a more general setting. It may make the paper seem weak. How do you think? Adding the open question or doing nothing are two options and the other is to change the title to ``The need for QRAM in quantum lattice sieving''?}
% \minki{I changed some parts of the introduction; I barely changed Section 1.1.}

% \minki{Our paper is roughly 29 pages. Could you suggest which part to reduce?}
% \jh{I excluded the conclusion section. Also, I commented out the proof of Lemma 1 and Lemma 2 in section 3, and added the short paragraphs for proof idea to defer the detailed proof to the Appendix. If it looks good, we need to cut about only 1 page.}
% \jh{Regarding 3) in email, height and width of Fig 1, Fig 3 (color graphs) are adjusted, as well as the length of arrow in Fig 2. It helps to reduce a few lines...}

\noindent
One of the major impacts of quantum computing is to efficiently solve the integer factoring problem with Shor's algorithm \cite{Shor99}, threatening the currently used cryptographic schemes like the RSA cryptosystem~\cite{RSA78}.
As a countermeasure, NIST has started to standardize post-quantum cryptography (PQC) to replace the currently deployed ones.
Lattice-based cryptography is one of the most promising post-quantum cryptography candidates due to its provable security~\cite{Reg09} and efficiency. Several schemes~\cite{avanzi19crystals,chen19algorithm,fouque18falcon} have thus been selected to be standardized by NIST \cite{NIST}.

The post-quantum security of lattice-based cryptography has been studied extensively, e.g. in~\cite{APS15,AGV+17,Wun19,Ngu21}.
Many lattice cryptanalysis algorithms  
heavily rely on lattice reduction algorithms \cite{HG07,Alb17,AGV+17,CHHS19,EJK20,GJ21,HKLS22,Mat22,PS24}, which in turn reduce to solving the shortest vector problem (SVP); therefore the SVP algorithm often dominates the cost of the overall attack.
% Indeed, the SVP is a fundamental subroutine in lattice reduction algorithms. 
Classically, lattice sieving and enumeration are the most promising approaches to solving the SVP. For the SVP over a $d$-dimensional lattice, lattice sieving has time complexity $2^{O(d)}$ but requires an exponential memory $2^{O(d)}$ as well~\cite{AKS01,NV08}. On the other hand, enumeration algorithms only require a polynomial-size memory while having super-exponential complexity $2^{O(d\log d)}$~\cite{Poh81}. 
% \minki{I removed the emphasis on the practicality of sieving considering the comments. Also, I changed many parts below to encompass the lattice enumeration.}
% There has been a long debate about which is better in practice; lattice sieving is believed to be the current winner, thanks to recent improvements. We focus on the lattice sieving algorithm throughout this paper.

A series of works have shown how to obtain asymptotic quantum speedups for both lattice algorithms. 
For enumeration, asymptotic quadratic quantum speedup~\cite{ANS18} have been obtained using the quantum backtracking technique~\cite{Mon18}, positively answering the conjectured quadratic speedup in ~\cite{NTRUHRSS}. The complexity of quantum lattice sieving~\cite{Laa16,Heiser21,CL21,BCSS23} is more involved. Compared to the best classical time complexity $2^{0.2925d+o(d)}$ \cite{BDGL16}, the current best quantum time complexity is $2^{0.2563d+o(d)}$ which is achieved using (reusable) quantum walk techniques~\cite{BCSS23}.
The concrete complexity of the quantum enumeration and sieving has been explored in~\cite{BvJLN23} and~\cite{AGP+20} respectively, which demonstrate that quantum speedups for lattice cryptanalysis are achievable with full-fledged quantum computers. 
% \minki{Let me put the hidden below somewhere later}
% The PQC standardization conservatively estimates the security based on the assumption of the full-quadratic speed up over the classical counterpart. While none of the above algorithms achieve such a quadratic speedup, those works have offered theoretical insights into the post-quantum security of lattice cryptography.}

Our current understanding of quantum computer architectures, however, has raised many concerns about full-fledged, unbounded-depth, error-free quantum computers. Two notable constraints on quantum devices must be considered---the \emph{quantum circuit depth} and the \emph{quantum random access memory (QRAM)}. Since qubits suffer from physical errors and decoherence, achieving large quantum circuit depth could be challenging. The limited quantum circuit depth model has thus been suggested by NIST~\cite[Section 4.A.5]{NISTSUB} in their standardization procedure.
%\jh{It seems that the "cite\{NISTSUB\}" can also be used instead of "cite\{NIST\}" at the end of 3 paragraphs above(which cites the website). Please check and if \{NIST\} must be preserved, update its access date please.}
%\minki{I think their purposes are different so preserving it may be better. I don't think we need to update the access date.}

QRAM is a quantum variant of classical RAM, allowing coherent access operation 
\[\ket{i}\ket{0}\mapsto \ket{i}\ket{x_i}\]
given a stored list $L=(x_i)_{i\in I}$ of classical data.
This operation requires coherent access to all data in $L$ and thus relies on a new quantum architecture. Many architectures~\cite{GLM08,GLM08b,MM16} have been suggested. However, all known proposals suffer from some drawbacks.
%\jh{Better to merge two sentences?- Although many architectures of QRAM~\cite{GLM08,GLM08b,MM16} have been suggested, all known proposals suffer from some drawbacks.} 
%\minki{While I usually write long sentences, I think shorter ones are better.}
\cite{Jaq23} discusses some fundamental limitations on these infrastructures. Based on the conclusion of \cite{Jaq23}, QRAM should be considered expensive or potentially unrealistic. Ideally, we would like to remove the use of QRAM in quantum algorithms to avoid the above issues. This is, for example, 
the case of the quantum collision finding problem where a quantum speedup without QRAM was achieved in \cite{CNS17}.

% Two notable constraints on quantum devices are the \emph{quantum-accessible classical memory (QRAM)} and the \emph{quantum depth} besides the physical assumptions. There is no reasonable suggestion for the efficient QRAM architecture~\cite{Jaq23} so far. Achieving large quantum depth could be challenging because of the decoherence. The limited quantum depth model has also been suggested by NIST~\cite[Section 4.A.5]{NISTSUB} in their standardization.\footnote{See \minki{Somewhere} for a more detailed discussion.}
\medskip

The quantum circuit depths and/or QRAM may be a hurdle for the quantum lattice algorithms. The quantum sieving estimation~\cite{AGP+20} shows that all known quantum sieving algorithms require huge quantum depth and QRAM, leaving the question of whether large QRAM and/or large quantum circuit depth are fundamental barriers to achieving a quantum advantage.
% The known quantum lattice algorithms require large quantum depths and/or QRAM.
% The circuits in the estimation for quantum enumeration~\cite{BvJLN23} have huge quantum circuit depth. 
% Both QRAM size and quantum depth are large in the quantum lattice sieving estimation~\cite{AGP+20}.
% Estimations for both enumeration and sieving algorithms~\cite{BvJLN23,AGP+20} suggest that quantum speedups require huge quantum circuit depth. 
% The estimations for quantum sieving in~\cite{AGP+20} also suggest
% \yixin{clarify what suggest means here}\minki{Changed}
% that quantum speedups are limited in the bounded QRAM or quantum circuit depth setting.
% In turn, the authors asked whether large QRAM and/or large quantum circuit depth are fundamental barriers to achieving the quantum advantage.

This question is also relevant for enumeration where a recent work~\cite{BBTV24} explored quantum speedups in the bounded quantum circuit depth model and concluded that the current quantum enumeration techniques are unlikely to provide practical speedup in this setting.

In this paper, we mainly focus on the QRAM aspect of the question:
\begin{center}
    \emph{Does quantum lattice sieving necessarily require a large QRAM?
    % \\What about the quantum depth?
    }
\end{center}
% \minki{One concern is that, in my opinion, our result in a sense is more general than~\cite{BBTV24} as we capture general frameworks. But I don't know how to explain this upto this point clearly. Maybe we can say more details in the contribution section.}
Our main question is particularly interesting in the near future where small-sized quantum computers are available, but QRAM does not exist. Furthermore, if the answer is yes, and it turns out that efficient QRAM is unlikely to exist and large quantum circuit depth is hard to realize, then the implication combining~\cite{BBTV24} could be even more striking---\emph{there could be no quantum speedup in lattice cryptanalysis!}
% If the quantum depth is also limited, 
% Of course, this is obviously false; for example, the quantum enumeration algorithm does not require QRAM~\cite{ANS18}. Still, in this hypothetical world, a fairly new idea is likely to be essential for the improvement of quantum time complexity in lattice cryptanalysis.

We would like to stress that quantum lattice sieving algorithms have been continuously improving since \cite{AGP+20}, and the QRAM-less variant has not been explored yet. This is in contrast with the classical setting where
much more is known. Kirshanova and Laarhoven~\cite{KL21} proved the tightness of the best classical lattice sieving~\cite{BDGL16} in the nearest-neighbor model that encompasses most sieving algorithms. 
% Similarly, the enumeration with the cylinder pruning~\cite{GNR10} is proven to be optimal~\cite{ANSS18} among a wide range of extreme pruning methods.
To the best of our knowledge, no similar results exist in the quantum
setting.\footnote{\cite{KL21} claims that the quantum
algorithm of \cite{Laa16}, which uses Grover's algorithm, is optimal. However, \cite{CL21}, which uses quantum walks, is already better than \cite{Laa16}.}

\subsection{This work}
%\jh{For the paragraphs in this subsection, I think the order of paragraphs are different to the order of contents in Section 3. Is is okay to keep this?}
%\minki{I think it's okay. Lower bound is likely to directly answer the main question, but it's hard to directly grasp the idea/formalization behind the lower bound without knowing algorithms}
In this paper, we study the role of QRAM\footnote{This paper focuses on the QRAM structure accessing classical data, usually referred to as QRACM. We refer to~\Cref {sec:qram} or~\cite{Jaq23} for more discussion on QRAM for quantum data.} in quantum lattice sieving. We provide some evidence that QRAM access is essential for quantum speedups and explore the fine-grained trade-off between the size of QRAM and the time complexity of quantum lattice sieving.

% In fact, we challenge these estimations by arguing the need for \emph{strong} QRAM in the quantum lattice sieving. \minki{I think these two lines could be aggresive; do you think it's okay?}
% \yixin{I think it's too aggressive, maybe rephrase say that we ``argue'' something?}

\paragraph{Need for (large) QRAM in quantum lattice sieving.}
We introduce a new quantum and classical lattice sieving (or near-neighbor) model that encompasses all known sieving algorithms~\cite{CL21,BCSS23,Heiser21,Laa16}. This model is similar to the one used in the classical sieving lower bound~\cite{KL21}. We then show that, in this model, any quantum algorithm beyond the classical lower bound must use a QRAM.
Putting it differently, our lower bound suggests that the QRAM-less quantum speedup for lattice sieving requires a fundamentally new idea.
% \minki{last sentence added}

% \yixin{Suggest: We introduce a lattice sieving model that encompasses all known quantum algorithms. This model is similar to the one used in the classical sieving lower bound~\cite{KL21}. We then show that,
% in this model, any algorithm must use a QRAM.}
% Our first contribution is to prove that the quantum speedup in lattice sieving is unlikely without using the QRAM in a lattice sieving model encompassing all known quantum algorithms, similar to one used in the classical sieving lower bound~\cite{KL21}. 
In fact, we prove a more general lower bound, which implies that a stronger quantum speedup in lattice sieving requires a larger QRAM. While our trade-off bound is not tight, it asserts the need for a somewhat large QRAM to achieve the complexity as in the current best quantum lattice sieving algorithm.

We note that our model includes two types of strategies to include the two algorithms using quantum walks~\cite{CL21,BCSS23}, which slightly deviate from the framework of~\cite{KL21}\footnote{Though their classical lower bound still applies (as we showed).}.
For the lower bound regarding the strategy including~\cite{CL21,BCSS23}, we make some reasonable assumptions about the use of QRAM: Roughly, it says that the efficient generation of coherent states of the vectors can only be achieved through the QRAM access.
% \minki{Is there any better way to describe the assumption?}

% \minki{Will write about the counterexample and some discussion}

\paragraph{Quantum lattice sieving with small QRAM.} 
Given the requirement of QRAM in the quantum speedup, we explore the fine-grained trade-off between the size of QRAM and the quantum complexity of lattice sieving. We revisit the trade-off algorithms that smoothly interpolate between the known classical and quantum complexity of lattice sieving, which is also briefly discussed in~\cite{CL21}. 
% \minki{Say something about the lower bound?}

% \minki{I noticed that the paragraph below is untrue... will modify}\yixin{Why is it untrue?}\minki{Heiser part is better than~\cite{CL21} in some param regime}
The graph showing the overall trade-off can be found in~\Cref{fig: quantum_lsf_sieving_result}. 
While the trade-off is mostly unchanged from~\cite{CL21} in the (very) small QRAM regime, we found that the bounded QRAM version of~\cite{Heiser21} gives a better time complexity than the one suggested in~\cite{CL21} for a moderate size QRAM.
These trade-offs are obtained using variants of Grover's algorithm with a bounded QRAM. We show that further improvements require a new way to use the QRAM by proving the optimality of some subroutines.\footnote{See \Cref{thm: lower_bound bounded QRAM search}. We are not aware of a similar lower bound in the literature.}

Another important caveat we found is that these trade-off algorithms require a strong assumption on the QRAM: The stored classical data must be replaceable very efficiently, i.e. in $2^{o(s)}$ time for a QRAM of size $2^s$. This efficient operation potentially forces a very specific implementation of QRAM or asks for multiple QRAMs to virtually implement such an operation. We refer to~\Cref{sec:discussion_QRAM} for a more detailed discussion.
We are not aware of any way to use small QRAM in lattice sieving without this assumption. This indicates that the quantum speedup for lattice sieving is challenging even with a moderate size QRAM.

These findings prompt us to revisit the comparison between enumeration
and sieving for the famous lattice reduction algorithm BKZ \cite{Schnorr87}.
If we assume that arbitrary depth quantum circuits are available
but we cannot build large QRAM, then quantum enumeration will outperform
quantum sieving even for relatively large dimension. 
\Cref{fig:sieving_vs_enum} plots the complexity of BKZ using quantum
enumeration \cite[Fig.~10]{ABFKSW20} (assuming a full quadratic speedup of \cite{ANS18}) compared to our lower bound on sieving with no QRAM
of \Cref{sec:lower}. We also include the best quantum algorithm for BKZ
with sieving, assuming no constraint on the QRAM size \cite{BCSS23}.
Note that all algorithms above achieve the same root Hermite factor
$k^{1/k}$, and that we have neglected some polynomial factors.
The goal of this comparison is therefore not to give exact values for
the cross-over points between enumeration and sieving but rather to give
orders of magnitude and to observe the impact of the (lack of) QRAM.

\begin{figure}
    \begin{center}
    \begin{tikzpicture}[scale=0.75, every node/.style={transform shape}]
        \begin{axis}[
            xlabel={block-size $k$},
            ylabel={$\tfrac{1}{k}\log_2(\text{complexity})$},
            height=5cm,
            width=12cm,
            domain=70:700,
            legend pos=south east,
            legend cell align=left
        ]
        \addplot[blue,variable=k] {(k/8*ln(k)/ln(2)-0.547*k+10.4)/2/k};
        \addlegendentry{BKZ+quantum enumeration}
        \addplot[red,variable=k] {0.2925};
        \addlegendentry{BKZ+sieving [no QRAM]}
        \addplot[green,variable=k] {0.2563};
        \addlegendentry{BKZ+sieving [full QRAM]}
        \end{axis}
    \end{tikzpicture}
    \end{center}
    \caption{\label{fig:sieving_vs_enum} Comparison between BKZ using
    quantum enumeration with full quadratic speedup, and BKZ using
    our lower bound on quantum sieving with no QRAM. We also include
    the best quantum algorithm with no constraints on the QRAM.}
\end{figure}
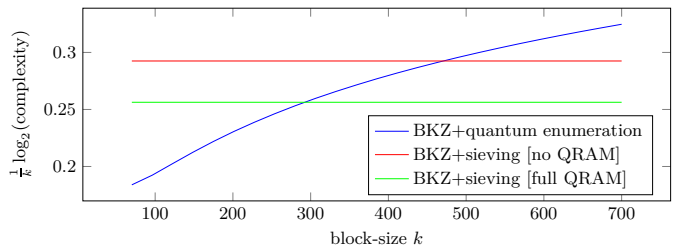

\paragraph{Symmetric key cryptography.}
The idea of designing quantum algorithms with a small QRAM can be applied to the collision finding and multi-target preimage search problems, discussed in~\cite{CNS17} without QRAM. 
% For example,  collision finding can be solved with quantum algorithms with much better complexity than in the classical case, if given access to a QRAM. Collision finding is an important problem for symmetric key cryptography and cryptographic hash functions. 
We demonstrate a smooth trade-off for the quantum collision finding problem between the results of \cite{BHT98} and \cite{CNS17}, which correspond to the maximal and minimal QRAM, respectively. We also show a similar trade-off for the multi-target preimage search problem bridging the QRAM-based multi-target Grover algorithm and the result in~\cite{CNS17}.

These results readily apply to the applications discussed in~\cite{CNS17}: hash collision finding, multi-user security, CBC mode of operations, as well as building blocks for advanced cryptanalysis such as~\cite{KLLN16}. As in the lattice cryptanalysis discussion, small practical QRAM helps for symmetric-key cryptanalysis but the same caveat applies regarding the need for efficient data replacement.

\paragraph{Quantum lattice sieving without QRAM.} 
Finally, we present, inspired by~\cite{Heiser21}, a quantum sieving algorithm without QRAM that has a depth complexity smaller than the classical sieving algorithms. 
This does not contradict our lower bound because it uses exponentially many qubits to operate many gates in parallel, which may be considered unrealistic; when $2^{0.207d}$ qubits are available, our algorithm runs in time about $2^{0.279d}$. To the best of our knowledge, this is the first QRAM-less quantum algorithm in the standard circuit model, faster than classical, with the minimal space of $2^{0.207d}$ for lattice sieving.\footnote{\cite{KMPM19} suggested a fast quantum algorithm without QRAM, but their algorithm works in a stronger model of distributed quantum computing~\cite{BGG+13}.} 
This result rules out a depth-alone lower bound for sieving, suggesting that the depth-width cost could be more desirable than the depth cost.
% \minki{This result rules out a similar depth lower bound to ours, suggesting that the depth-width cost could be more desirable than the depth cost. (I want to say sth similar but don't know how to phrase well...)}
% \jh{How about "This does not contradict our l.b. because it uses exponentially many qubits to do many operations in parallel, which may be considered unrealistic; when ...."}
% \minki{Changed}

% \input{1.z.overview}

\bigskip

\paragraph{Organization.} This paper is organized as follows. 
\Cref{section: preliminaries} presents some preliminaries required for this paper. 
We present the quantum algorithms using small QRAM in~\Cref{section: limited_qracm_grover}. 
As applications of these algorithms, the trade-off algorithms between the QRAM size and time complexity are given in~\Cref{section: quantum_lsf_sieving} for lattice sieving and in~\Cref{section: app_to_symmetric_key} for symmetric key cryptography. 
Our lattice sieving model and lower bounds are presented in~\Cref{sec:lower}, with a more formal treatment in~\Cref{app: lowerbound}. 
Finally, \Cref{sec:withoutQRACM} shows how to solve the sieving problem without a QRAM but at the expense
of using many qubits. 
% We discuss the implications of this result on the lower bound obtained previously.
% \minki{section 7 description}
% \yixin{Section~\ref{sec:withoutQRACM} shows how to solve the sieving problem without a QRAM but at the expense
% of using many qubits. We discuss the implications of this result on the lower bound obtained previously.}

% An often overlooked point when designing quantum algorithms is the size of quantum random access memory (QRAM) used. QRAM is a quantum version of classical random access memory. While the number of qubits used is a well established measure of memory complexity, the situation is less clear when it comes to QRAM.
% Indeed, the feasibility of large QRAM remains controversial, since making a large QRAM that can store coherent data with high fidelity is technically challenging \cite{Jaq23}. Moreover, QRAM can be categorized into Quantumly Random Accessible Classical Memory (QRACM) and Quantumly Random Accessible Quantum Memory (QRAQM) depending on the type of data stored \cite{Kuperberg05}. It is generally considered that QRACM is a more plausible model than QRAQM. 
% In the above example, $T=2^{0.2653d+o(d)}$ in \cite{Laa16} and $T=2^{0.2571d+o(d)}$ in \cite{Heiser21} are achieved with QRACM, while quantum walk algorithms requires QRAQM to achieve $T=2^{0.2653d+o(d)}$ (\cite{CL21,BCSS23}).

\section{Preliminaries}
\label{section: preliminaries}

\paragraph{Notations.}
% \label{section: prelim_notation}
% \todojh{I only changed the title of subsection from "Geometry on the sphere" to "Notation"}
%\minki{For now it is a collection of random facts.}
We use bold lower-case letters to denote vectors and bold upper-case letters to denote matrices. The Euclidean norm of a vector $\vecv$ is denoted by $\|\vecv \|$. The inner product of two vectors is denoted by $\inner{\cdot,\cdot}$. We usually denote the dimension of the lattice by $d$, and also assume that the size of the vectors are polynomial in $d$.

When discussing the running time of the algorithms solving lattice problems, 
we assume that the vectors (used in the basis or in computations) have 
bit-size polynomial in the ambient dimension $d$). In particular, we will ignore the resulting polynomial factors due to manipulating vectors.
% slightly paraphrase(241002T1530+0) - When analyzing the running time of algorithms that solve lattice problems, we disregard polynomial factors related to the bit-length of the individual input basis vectors (i.e., we assume the bit-size of the input basis is polynomial in the ambient dimension $d$).

\subsection{Lattice}
\label{section: prelim_sieving}
A (full-rank) \emph{$d$-dimensional lattice $\mathcal{L}$} is a discrete additive subgroup of $\mathbb{R}^d$, whose elements can be uniquely expressed as the linear combinations of linearly independent basis vector $\mathbf{B}=\{ \vecb_1, \vecb_2, \cdots, \vecb_d \}$.
\begin{equation}
    \mathcal{L}(\mathbf{B}) = \left\{ \sum_{i=1}^{d}{x_i\cdot \vecb_{i}} \ | \ x_i \in \mathbb{Z}, \vecb_i \in \mathbf{B} \right\}
    \label{eq: def_lattice}
\end{equation}
% Basis of a lattice are not unique, and $\mathbf{B}, \mathbf{B}'$ are equivalent if $\mathbf{B}=\mathbf{B}'U$ for some unimodular matrix $U \in \mathbb{Z}^{n\times n}$. 
Given a basis $\mathbf{B}$ of the lattice,
the \emph{shortest vector problem (SVP)} asks
% is a computational problem defined on the lattice. 
% Given a basis $\mathbf{B}$ of the lattice, the goal of SVP is 
to find the vector $\Vec{s}\in \mathcal{L}(\mathbf{B})$ such that $\norm{\Vec{s}}=\lambda_1(\mathcal{L}):=\min_{\vecv\in \mathcal{L}(\mathbf{B})-\{\mathbf{0}\}}{\norm{\vecv}}$.

\subsection{Quantum computing}
We consider the quantum circuit model consisting of single- and two-qubit gates without any locality constraint. We further assume that the number of qubits available to algorithms is bounded by some polynomial, except for~\Cref{sec:lower,sec:withoutQRACM} which focus on theoretic perspectives.

We count the depth of the quantum circuit as the time complexity.
We occasionally ignore the polynomial factors and only focus on the exponential terms when discussing the complexity.
Since this paper mostly assumes a small number of qubits, choosing different complexity measures does not change the results of this paper much. We explicitly describe the number of qubits when we discuss algorithms with large qubits.
For a more detailed introduction to quantum computing, we refer to~\cite{NC10}.
% \minki{Actually, this is tricky, because the counterexample section has exponential quantum memory..}

\paragraph{Quantum oracles.} 
Let $X$ be a set.
For a function $f:X \to \{0,1\}^m$, the oracle $O_f$ is a unitary that computes:
\begin{align}
    O_f:\ket{x,y} \mapsto \ket{x,y\oplus f(x)}.
\end{align}
The circuits for computing most oracles in this paper will be explicitly given. The time complexity of the oracle is the same as the corresponding circuit.

For a projector $P$ acting on the span of $X$, we similarly define the projection oracle $O_P$ by:
\begin{align}
    O_P:\ket{\psi}\ket{b}\mapsto\begin{cases}
        \ket{\psi}\ket{b\oplus 1}&\text{ if }\ket{\psi}\in Im(P),\\
        \ket{\psi}\ket{b}&\text{ if }\ket{\psi}\in Ker(P).
    \end{cases}
\end{align}

We occasionally use quantum amplitude amplification as a subroutine in our algorithms, which can be seen as a generalization of Grover's algorithm.
\begin{theorem}[{Quantum amplitude amplification~\cite{BHMT02}}]
    \label{thm:QAA}
    Let $P$ be a projector acting on the span of $X$.
    Let $Init$ be a quantum algorithm that generates $\ket{\phi}=\alpha\ket{\phi_P}+\beta\ket{\phi_P^\bot}$, where $\ket{\phi_P} \in Im(P)$ and $\ket{\phi_P^\bot}\in Ker(P)$. Let $\theta\in[0,\pi/2]$ be such that $\sin \theta = |\alpha|$. Let $N=\lfloor \frac{\pi}{4\theta}-\frac{1}{2}\rfloor$. 
    If the time complexities of $O_P$ and $Init$ are $T_P,T_{Init}$, respectively, then 
    there exists a quantum algorithm that produces a quantum state sufficiently close\footnote{This can be done in a standard way, e.g., as in~\cite{BBHT98}.} to $\ket{\phi_P}$ in time $O(N(T_P+T_{Init})).$ If $\alpha =o(1)$, it can be written as $O((T_P+T_{Init})/\alpha).$
\end{theorem}

If $X=\{0,1\}^n$ and the projector $P$ is defined by a function $f:\{0,1\}^n \to \{0,1\}$, we recover Grover algorithm. In particular, it finds $x\in \{0,1\}^n$ such that $f(x)=1$ in time $O(\sqrt{2^n/|f^{-1}(1)|}).$

\begin{theorem}[{Quantum Minimum Finding Algorithm~\cite[Theorem 1]{DH96}}]
    Given a list of $N=2^n$ values from an ordered set, there is an algorithm to find the index of the minimum value with probability at least 1/2 in time $O(\sqrt{N})$. 
\end{theorem}

\subsection{Quantum random access memory}
\label{sec:qram}

We consider \emph{quantum random access memory (QRAM)} in addition to the quantum circuit model. 
The main feature of the QRAM is that it allows queries that are superposed in addresses, thereby allowing the superposition of desired data stored in memory. 

More precisely, using the QRAM storing $m$ elements $L=(x_1,...,x_{m})$, a quantum algorithm can apply a special QRAM gate that works as follows in the computational basis:
\begin{align}\label{eqn:QRAMdef}
    U_{QRAM}^{L}:\ket{i,y} \mapsto \ket{i,y \oplus x_i}.
\end{align}
    % \jh{I wonder whether it is better to write $U_{QRAM}(L)$ instead of put $(L)$ into subscript}\minki{I is never used in elsewhere so feel free to modify as your favorite}
The quantum algorithms in this paper are specified by a quantum circuit with single-, two-qubits and QRAM gates along with the memory storing the data in the QRAM. We count the QRAM gate as a single gate in the time complexity.

We assume that the memory used in the QRAM gate always stores classical data. Precisely, we only consider \emph{quantum random access classical memory (QRACM)}. One may consider quantum random access quantum memory (QRAQM), which is more general~\cite{Jaq23}. 
Given that maintaining a coherent state for a long time is hard, restricting the algorithms to QRACM only and having small quantum memory is a reasonable model: The small quantum memory is the only part that maintains quantum information perpetually, while the other parts only become coherent ephemerally.
% \bg{How about replacing 'small quantum memory' to 'small number of qubits'?}

Although there is no agreed model of the QRA(C)M, writing data on it will require non-trivial cost which is at least logarithmic in the memory size (and more likely to require polynomial or exponential time). However, we assume that writing the data to, or overwriting (i.e. reset and load another data) the QRACM can be done in time $O(1)$.
%While this assumption is unrealistically strong, it yields a natural trade-off for some quantum algorithm using QRACM, as can be seen in Section~\ref{section: quantum_lsf_sieving}. 
%It also gives a lower bound on the quantum algorithm, counting only the cost of oracles queries.
% \minki{I didn't touch below yet.}
While this assumption is very strong, the complexity assuming this assumption gives a lower bound of the complexity of the algorithm, in the sense that the cost related to memory I/O is excluded, and only the number of memory accesses is taken into account. Also, the result yields a natural trade-off for some quantum algorithms using QRAM, as can be seen in~\Cref{section: quantum_lsf_sieving}. More discussion on this assumption can be
found in~\Cref{sec:discussion_QRAM} where we argue that our model is not
as unrealistic as it appears, if we think about QRAMs as physical
objects.
\section{Quantum Algorithms with Bounded QRAM}
\label{section: limited_qracm_grover}
\def\charindex{1}

% \minki{I feel some terms in this section (and Section 5) seem ugly. How about changing $\sqrt{M}$ instead of $2^{m/2}$, say?}
% \bg{Considering the terms such as $\frac{\sqrt{S}q}{M}$ in Theorem 3, It would be better to replace $2^m$ and $2^s$ to $M$ and $S$, respectively.. }
% \minki{I will change this.}
This section presents the basic quantum algorithms using bounded-size QRAM. We suppose that a set 
\ifnum\charindex=1
$X=\{x_1,...,x_{M}\}$ of size $M$
\else
$X=\{x_1,...,x_{2^m}\}$ of size $2^m$
\fi
is given to the algorithm in the memory, but the algorithm is allowed to use QRAM of size 
\ifnum\charindex=1
$S\le M$,
which means the algorithm can coherently access at most $S$
elements in $X$.
For simplicity, we assume that $M$ is divided by $S$.
\else
$2^s\le 2^m$,
which means the algorithm can coherently access at most $2^s$
elements in $X$.
\fi 
Looking ahead, we are mainly interested in the case where $X$ is a random subset of $\Sd$ for later use in lattice sieving, but the results in this section hold regardless of the choice of $X$.

Consider the function $f:X \to \{0,1\}$ and the problem of finding $x\in X$ such that $f(x)=1$\footnote{With the promise that such an $x$ exists with a high probability.}.
When 
\ifnum\charindex=1
$X=[M]$ is explicitly given and the function $f$ can be computed by a circuit (or given as an oracle), it is well known that 
Grover algorithm requires only $O(\sqrt M)$
\else
$X=\{0,1\}^m$ is explicitly given and the function $f$ can be computed by a circuit (or given as an oracle), it is well known that 
Grover algorithm requires only $O(2^{m/2})$
\fi
computations of $f$ to solve this problem~\cite{Grover96}.
% If the oracle $\mathcal{O}_f$ is efficiently implementable as a quantum circuit, no additional quantum devices are required. 

The situation becomes more complicated if $X$ is not explicitly specified a priori. To execute the quantum search or the amplitude amplification, we need to implement the reflection map
\ifnum\charindex=1
\begin{align}\label{eqn:reflection}
    \text{Ref} = I-2\ketbra{\psi}\text{ for }\ket{\psi} = 
    \sum_{i \in [M]} \frac{\ket{i,x_i}}{\sqrt M}.
\end{align}
\else
\begin{align}\label{eqn:reflection}
    \text{Ref} = I-2\ketbra{\psi}\text{ for }\ket{\psi} = 
    \sum_{i \in [2^m]} \frac{\ket{i,x_i}}{2^{m/2}}.
\end{align}
\fi
% \minki{I notice that it is probably easier to apply $I-2\ketbra{\phi}$ for $\ket{\phi}=\sum_{i\in[|X|]} \ket{i,\vecx_i}/\sqrt{|X|}$ in practice because constructing $\psi$ above is slightly complicated... }
If there is no size bound on the QRAM, an algorithm that outputs $\ket{\psi}$ can be efficiently implemented as follows
\ifnum\charindex=1
\begin{align}\label{eqn: initwithfullQRAM}
    \ket{0} \mapsto \sum_{i\in [M]}\frac{\ket{i,0}}{\sqrt M} \mapsto \sum_{i\in [M]}\frac{\ket{i,x_i}}{\sqrt M}
\end{align}
\else
\begin{align}\label{eqn: initwithfullQRAM}
    \ket{0} \mapsto \sum_{i\in [2^m]}\frac{\ket{i,0}}{2^{m/2}} \mapsto \sum_{i\in [2^m]}\frac{\ket{i,x_i}}{2^{m/2}}
\end{align}
\fi
where we use the QRAM gate (\cref{eqn:QRAMdef}) in the second map, and the amplitude amplification gives the same complexity as in Grover's algorithm.
When the algorithm is limited to a QRAM of size at most 
\ifnum\charindex=1
$S$,
\else
$2^s$,
\fi
the following lemma shows that we can obtain a slightly worse quantum speedup. Setting $S=1$ and $S=M$ give the classical exhaustive search algorithm and Grover algorithm, respectively.
% can be described by~\Cref{lemma: limited_qracm_grover} for $S=1$ and $S=M$, respectively.

\ifnum\charindex=1
\begin{lemma}
    Let $X$ be a set of size $M$ stored in a classical memory of the algorithm.
    Assume that the function $f: X\rightarrow \{0, 1\}$ is randomly chosen so that $\Pr[f(x)=1]=c/M$ holds for each $x\in X$ independently for some $c\ge 6$.\footnote{It ensures the existence of $x^*$ such that $f(x^*)=1$ with probability at least 0.995.}
    %promised that there exists a solution $x \in X$ such that $f(x)=1$. 
    % Assume that all the elements of $X$ are stored in a classical list. 
    If QRAM of size $S$ is available for some $0\le S \le M$, then there exists an algorithm $A$ that can find $x^*$ such that $f(x^*)=1$ in $O\left(\frac{M }{\sqrt S}\right)$ evaluations of $f$ with a sufficiently high probability (say 0.99).
    \label{lemma: limited_qracm_grover}
\end{lemma}

% \jh{added proof sketch}
\begin{proof} [Proof (sketch)]
The proof idea is that by storing $S$ elements of $X$ in the QRAM and searching the solution in those elements only using $O(\sqrt{S})$ evaluations. Repeating the procedure $M/S$ times for each block, $O\left(\frac{M}{\sqrt{S}}\right)$ evaluations are required. The formal proof is given in~\Cref{app: omitted proofs in QRAM bdd}.
\end{proof}

\else
\begin{lemma}
    Let $X$ be a set of size $2^m$ stored in a classical memory of the algorithm.
    Assume that the function $f: X\rightarrow \{0, 1\}$, promised that there exists a solution $x \in X$ such that $f(x)=1$, is computable in time $T$. 
    % Assume that all the elements of $X$ are stored in a classical list. 
    If QRAM of size $2^s$ is available for some $0\le s \le m$, then there exists an algorithm $A$ that can find the solution in time $O(T\cdot 2^{m-s/2})$ with a sufficiently high probability (say 0.99).
    \label{lemma: limited_qracm_grover}
\end{lemma}

%\jh{added proof sketch}
\begin{proof} [Proof (sketch)]
The proof idea is that by storing $2^s$ elements of $X$ in the QRAM and searching the solution in those elements only, using $O(2^{s/2})$ evaluations. Repeating the procedure $2^{m-s}$ times for each block, $O(2^{m-s/2})$ evaluations are required. The formal proof is given in the~\Cref{app: omitted proofs in QRAM bdd}.
\end{proof}

The classical exhaustive search algorithm and Grover algorithm can be described by~\Cref{lemma: limited_qracm_grover} for $s=0$ and $s=m$, respectively.
% The result of Lemma~\ref{lemma: limited_qracm_grover} includes the complexity of quantum search ($O(2^{n/2})$ with $s=n$), as well as the classical one ($O(2^{n})$ with $s=0$).\\
\fi

We remark that the assumption for the efficient replacement of QRAM is crucial in the above lemma. Otherwise, the first step to store elements in QRAM can dominate the algorithm's complexity.

We also consider the following generalization of the search problem: Given two sets $X$ and $Y$ and a function $f:X \times Y \to \{0,1\}$, find most solution pairs $(x, y) \in X\times Y$ such that $f(x,y)=1$. 
% We also consider the bounded QRAM case, but we assume that two QRAMs are available, which will be used for $X$ and $Y$, respectively. 
\Cref{lemma: limited_qracm_grover_pair} shows the complexity of this problem in the bounded QRAM case.

\ifnum\charindex=1
\begin{lemma}
    Let $X, Y$ be a set of sizes $M_1, M_2$ respectively, stored in a classical memory of the algorithm. 
    Assume that the function $f: X \times Y \rightarrow \{0, 1\}$ has $K$ (uniformly distributed) solutions $(x, y)\in X\times Y$ such that $f(x,y)=1$. 
    If two QRAMs of size $S$ each are available for some $0 \le S \le \max(M_1,M_2)$, then there exists an algorithm $A'$ that can find $\Omega(K)$ solutions using
    \begin{enumerate}
        \item $O\left (\sqrt{M_1 \cdot M_2 \cdot K}\right )$ evaluations of $f$ if $\sqrt{\frac{M_1\cdot M_2}{K}}\le S \le \max(M_1, M_2)$ and 
        \item $O\left( \frac{M_1 \cdot M_2}{S}\right)$ evaluations of $f$ if $1 \le S \le  \sqrt{\frac{M_1\cdot M_2}{K}}$.
    \end{enumerate}
    \label{lemma: limited_qracm_grover_pair}
\end{lemma}
% \jh{added proof sketch}
\begin{proof} [Proof (sketch)]
By storing $S$ elements of $X$ and $Y$ in separate QRAMs, superposition of $S^2$ possible pairs can be generated, and the amplitude amplification can be applied, similar to the proof in \Cref{lemma: limited_qracm_grover}. The result slightly differs because we need to find $\Omega(K)$ solutions. Detailed proof is given in~\Cref{app: omitted proofs in QRAM bdd}.
\end{proof}

\else
\begin{lemma}
    Let $X, Y$ be a set of size $2^{m_1}, 2^{m_2}$ respectively, stored in a classical memory of the algorithm. Assume that the function $f: X \times Y \rightarrow \{0, 1\}$ has $2^k$ solutions $(x, y)\in X\times Y$ such that $f(x, y)=1$ is computable in time $T$. If two QRAMs of size $2^s$ each are available for some $0 \le s \le \max(m_1, m_2)$, then there exists an algorithm $A'$ that can find $2^k$ solutions in time \textbf{(1)} $O(T\cdot 2^{(m_1+m_2+k)/2})$ if $(m_1+m_2-k)/2 \le s \le \max(m_1, m_2)$ and \textbf{(2)} $O(T\cdot 2^{m_1+m_2-s})$ if $0 \le s < (m_1+m_2-k)/2$.
    \label{lemma: limited_qracm_grover_pair}
\end{lemma}

%\jh{added proof sketch}
\begin{proof} [Proof (sketch)]
By storing $2^s$ elements of $X$ and $Y$ in separate QRAMs, superposition of $2^{2s}$ possible pairs can be generated, and the amplitude amplification can be applied, similar to the proof in \Cref{lemma: limited_qracm_grover}. The result slightly differs because we need to find $\Omega(2^k)$ solutions. Detailed proof is given in~\Cref{app: omitted proofs in QRAM bdd}.
\end{proof}

\fi

We also prove the lower bound corresponding to~\Cref{lemma: limited_qracm_grover}, albeit in a slightly restricted model. The condition ``each query contains at most $S$ elements'' captures the situation that whenever $A$ makes a QRAM query, $A$ evaluates the $f$ information of the QRAM entries.
The proof is more involved and uses a variant of the compressed oracle argument~\cite{Zha19} for the Bernoulli random functions~\cite{CGKSW23}. We defer the proof to~\Cref{subsec: QRAM bound search proof}.

\begin{theorem}\label{thm: lower_bound bounded QRAM search}
    Let $X$ be a set of size $M$. Suppose that $f:X\to \{0,1\}$ be a random function where $f(x)=1$ holds with probability $p=\Omega(1/M)$ for each $x\in X$ independently. Let $A$ be a quantum algorithm that makes at most $q$ queries to the quantum oracle access to $f$. If each query of $A$ to $f$ contains at most $S$ elements in $X$ in the computational basis, then it holds that 
    \[
    \Pr\left[A^f \to x : f(x)=1\right] = O\left(\sqrt S \cdot pq\right).
    \]
    In particular, any algorithm finding the solution $x$ requires making at least $\frac {M}{\sqrt S}$ queries to $f$ for $p=\Theta(1/M)$.
\end{theorem}
\section{Time-QRAM Trade-off for Quantum Lattice Sieving}
\label{section: quantum_lsf_sieving}
This section presents the fine-grained trade-off between time and the QRAM size in quantum lattice sieving based on the algorithms in~\Cref{section: limited_qracm_grover}. We focus on the locality-sensitive filtering(LSF) method in~\cite{BDGL16}, which is proven to be optimal among a wide class of sieving algorithms in the classical setting~\cite{KL21}.

Locality-sensitive filtering defines a family of functions called \textit{filters} $\{\mathcal{F}_i\}_{i\in [t]}$. A vector $\vecv$ \textit{passes} a filter $\mathcal{F}_i$ if (say) $\mathcal{F}_i(\vecv)=1$, otherwise $\mathcal{F}_i(\vecv)=0$. With careful construction of filter family, the key part of LSF is that one can find all filters that a vector passes, without calculating every $\mathcal{F}_i(\vecv)$ for $i\in [t]$.
% \minki{Introduce LSF here, with some parameterized descriptions. We also need to describe a little about quantum LSF in Heiser.}\bg{Should the "spherical lsf" paragraph below Lemma 4.1 move here? I added a brief description of Heiser.}

\subsection{Sieving with locality-sensitive filtering}
Before going into the quantum algorithm, we review the classical lattice sieving algorithm with locality-sensitive filtering in \Cref{alg: NV_sieve}.

\paragraph{Heuristic lattice sieving.}
% The outline of the lattice sieving is summarized in~\Cref{alg: NV_sieve}. While provable lattice sieving is inefficient,~
\cite{NV08} proves that $|L|=(\frac{4}{3})^{d/2 + o(d)} $ suffices to heuristically solve the shortest vector problem in the lattice of dimension $d$. Throughout this paper, we fix the size of the input list as
\begin{align}\label{eq: n}
    n:=|L|\approx2^{0.2075d + o(d)}.
\end{align}

\begin{algorithm}
\caption{Lattice Sieving Algorithms}\label{alg: NV_sieve}
\begin{algorithmic}[1]
\Require \textbf{Input} $R,R'=\gamma \cdot R>0$, $L$: list of input vectors with $\|\vecv\| \le R\ \ \  \forall \vecv \in L$ 
%, $R, R'$: upper bound of input and output vectors, respectively
\Require \textbf{Output} $L'$: list of output vectors with $\norm{\vecv'}\le R'\ \ \ \forall \vecv' \in L'$
\medskip

% \Procedure{\textsc{Sieving}}{$L, R'$}
    \State $L'\gets \emptyset$
    %Initialize the output list $L' \leftarrow \empty$ and center list $C \leftarrow \{\Vec{0} \}$ \\
    \ForAll{$\vecv \in L$}
        \If{$\exists \vecw\in L \text{ such that } \vecv \neq \vecw \text{ and } \norm{\vecv - \vecw}\le R'$}\label{item:sieve_comp}
            \State $L' \gets L' \cup \{\vecv - \vecw\}$\label{item:sieving_add}
        % \Else
        %     \State Add $\vecv \text{ to } C$
        \EndIf
    \EndFor
    \State \textbf{return }$L'$
% \EndProcedure
\end{algorithmic}
\end{algorithm}

We call that the vector $\vecv- \vecw$ added in~\Cref{item:sieving_add} the \emph{reduced vector}. The pair $(\vecv,\vecw)$ is called by the \emph{reducing} pair.
We always consider $R'/R\to 1$ (as $d\to \infty$). In this case, the condition $\|\vecv-\vecw\|\le R'$ is roughly equivalent to the angle between them is less than $\pi/3$, which we focus on.
The time complexity is dominated by the time to find the close pairs in~\Cref{item:sieve_comp}. 
In the original heuristic sieve algorithm~\cite{NV08}, it is done by the exhaustive search on $(\vecv,\vecw) \in L\times L$, thereby the time complexity becomes $n^2 \approx 2^{0.4150d + o(d)}$. Hash-based approaches significantly reduce the time complexity, and we review the most efficient algorithm based on the locality-sensitive filtering below.

\paragraph{Geometry of sphere.}
Consider the unit sphere $\Sd:= \{\vecx \in \R^d :\|\vecx\|=1\}$ and half-(hyper)spaces $\H_{\vecv,\alpha}:=\{\vecx \in \R^d:\inner{\vecx,\vecv}\ge \alpha\}$. For $\vecv,\vecw \in \Sd$ and $\alpha,\beta \in [0,1]$, spherical caps and wedges are defined by
\begin{align*}
    \cC_{\vecv,\alpha}
        := \Sd \cap \H_{\vec,\alpha},
    ~~~
    \cW_{\vecv,\alpha,\vecw,\beta}
        := \Sd\cap \H_{\vecv,\alpha}\cap\H_{\vecw,\beta}.
\end{align*}
Let $\mu$ be the canonical Lebesgue measure over $\R^d$.
The relative volumes of spherical caps and wedges have an important role in the lattice analysis. If the vectors $\vecv,\vecw\in\Sd$ satisfies $\inner{\vecv,\vecw}=\cos\theta$, we define
\begin{align*}
    \cC_d(\alpha):= 
        \frac{\mu(\cC_{\vecv,\alpha})}
        {\mu(\Sd)},
    ~~~
    \cW_d(\alpha,\beta,\theta):= 
        \frac{\mu(\cW_{\vecv,\alpha,\vecw,\beta})}
        {\mu(\Sd)}.
\end{align*}
We recall the estimates of these two values for large $d$~\cite{MV10,BDGL16}.
\begin{lemma}\label{lem: probs_spherecap}
    For arbitrary constants $\alpha,\beta \in (0,1)$ and $\theta \in [0,\pi]$, the following asymptotic formulas hold:
    \begin{enumerate}
        \item $\cC_d(\alpha)=
            \poly(d)\cdot \left(\sqrt{1-\alpha^2}\right)^d$, and
        \item $\cW_d(\alpha,\beta,\theta)=
            \poly(d) \cdot \left(\sqrt{1-\gamma^2}\right)^d$ 
            for $\gamma=\sqrt{\frac{\alpha^2 + \beta^2 - 2\alpha \beta \cos \theta}{\sin^2 \theta}}.$ \\
        In particular, if $\alpha=\beta$, 
            $\cW_d(\alpha,\beta,\theta)=
            \poly(d) \cdot \left(\sqrt{1-\frac{2\alpha^2}{1+\cos \theta}}\right)^d$.
    \end{enumerate}
\end{lemma}

\paragraph{Locality-sensitive filtering.}
The (spherical) locality-sensitive filtering is a family of functions called \emph{filters} $\{\mathcal F_{i,\alpha}\}_{i\in [t], \alpha\in (0,1)}$ that are specified by a vector in a set $\{\vecc_1,...,\vecc_t\}$ (sometimes called centers) in $\Sd$ along with the parameter $\alpha \in (0,1)$. 
This family has a nice property that, given a vector $\vecv \in \Sd$ and $\beta \in (0,1)$ as input, it is possible to find all indices $i\in [t]$ such that $\inner{\vecv,\vecc_i} \ge \beta$ efficiently. We call $\mathcal F_{i,\beta}$ by the relevant $\beta$-filters for these $i$'s. 
Formally, if the number of relevant $\beta$-filters is $f_{\vecv,\beta}$, then it finds all indices in time $O(f_{\vecv,\beta}).$ The construction of locality-sensitive filtering is based on the random product code, and the expected number of relevant filters is $t \cdot \cC_d(\beta)$.

\bigskip

The result of classical lattice sieving can be summarized as follows. We include a sketch of the proof for the completeness of the paper.

% the optimized/unoptimized time complexity of the classical LSF is given in the following theorem.

% Quantum LSF applies Grover search to the classical LSF algorithm. Before going into the quantum algorithm, the optimized/unoptimized time complexity of the classical LSF is given in the following theorem.

% \yixin{If I understand correctly, the $n$ in the statement
% is actually $n^2\approx2^{0.4150d+o(d)}$ defined page 5. This is
% really not obvious for a reader. I suggest creating an equation
% for $n$ somewhere and refering to it when we use it?} \bg{I added an equation for $n$ at preliminaries.}
% \minki{While I'm modifying the preliminary, the equation about $n$ is deleted. I think we can explain some computations here; will do tomorrow.}
\begin{theorem}[{\cite[Theorem 7.1]{BDGL16}}]\label{thm: classical sieving}
    % Let $R'=(1-\epsilon) R$ for a sufficiently small $\epsilon>0$.
    Let $L$ be a list of $n$ random input vectors sampled from $\Sd$ 
    where $n$ is defined in \eqref{eq: n}. 
    There is a classical lattice sieving algorithm, given $L$ and parameters $\alpha, \beta\in (0, 1)$ as input, that outputs $\Omega(n)$ reduced vectors based on the LSF with $t=\mathcal{W}_{d}(\alpha, \beta, \pi/3)^{-1}$ filters. The expected running time of the algorithm is given as follows:
    \begin{equation}
        T =
        \underbrace{nt \cdot \mathcal{C}_{d}(\beta)}_{\text{fill $\beta$-filters}}
        \;+ \;
        n\cdot\Bigg(
            \underbrace{t \cdot \mathcal{C}_{d}(\alpha)}_{\text{find $\alpha$-close filters}}
            +\underbrace{nt\cdot \mathcal{C}_{d}(\alpha)\cdot \mathcal{C}_{d}(\beta)}_{\text{examine vectors in $\alpha$-close filters}}
        \Bigg).
        \label{eq: classical_lsf_unopt}
    \end{equation}
    In particular, for the optimal choice,
    the time complexity becomes $2^{0.2925d+o(d)}.$
    \label{thm: classical_lsf}
\end{theorem}
\begin{proof}[Proof (sketch)]
    Divide $L=C \cup (L\setminus C)$ into two lists of similar sizes. 
    We first describe the algorithm $A$ as follows.
    \begin{enumerate}
        \item Prepare the spherical locality-sensitive filtering of size $t$ defined by a set $\{\vecc_1,...,\vecc_t\}$. Define $t$ empty lists $B_1,...,B_t$. Define an empty list $L'$.
        \item For each vector $\vecw \in C$, find all $i\in [t]$ such that $\mathcal F_{i,\beta}$ is the relevant $\beta$-filters of $\vecw$, and append $\vecw$ to $B_i$.
        \item For each $\vecv \in L\setminus C$, find all $i\in [t]$ and $\vecw \in B_i$ such that $\mathcal F_{i,\alpha}$ is the relevant $\alpha$-filters of $\vecv$, and check if $\|\vecv - \vecw\| \le R'$. If true, append $\vecv- \vecw$ to $L'$.
        \item Output $L'$.
    \end{enumerate}

    We first analyze the time complexity of the algorithm $A$. Note that $|C|,|L\setminus C| = O(n)$. The first step can be done efficiently. For the second step, finding the relevant vectors can be done efficiently due to the definition of locality-sensitive filtering. Since the expected number of the relevant filters is $t\cdot \cC_d(\beta)$, this step can be done in time $O(nt \cdot \cC_d(\beta))$ at total. The third step is similar; finding relevant $\alpha$-filters is done in time $t\cdot \cC_d(\alpha)$ and each of $B_i$ contains $n\cdot \cC_d(\beta)$ vectors on expectation, giving the desired time complexity as in \cref{eq: classical_lsf_unopt}.

    For the correctness of the algorithm, consider a close vector pair $(\vecv,\vecw)$, i.e., $\inner{\vecv,\vecw} \ge \cos{(\frac{\pi}{3})}$.
    The probability that a random vector $\vecc$ defines a relevant $\beta$- and $\alpha$-filter of $\vecv$ and $\vecw$, respectively, is precisely computed by $\cW_d(\alpha,\beta,\pi/3).$
    Hence, by~\Cref{lem: probs_spherecap}, choosing
    \begin{equation}
        t = \cW_d(\alpha,\beta,\theta)^{-1}=\cW_d\left(\alpha,\beta,\frac{\pi}{3}\right)^{-1} \approx \left(1-\frac{4}{3}\left(\alpha^2-\alpha\beta+\beta^2\right)\right)^{-d/2}
        \label{eq: filter_num}
    \end{equation}
    ensures that there exists a filter that can be used to find $\vecv-\vecw$ on expectation. Since $n=2^{0.2075d+o(d)}$ ensures that the number of such elements is at least $\Omega(n)$, we conclude that the output $L'$ of $A$ contains $\Omega(n)$ vectors.

    By setting the three terms to be equal to each other, the optimization results yield $\alpha=\beta=\frac{1}{2}$ and $T=t=\left( \frac{3}{2} \right)^{d/2+o(d)}=2^{0.2925d+o(d)}$.
\end{proof}

% \minki{Can we describe Thm 4.4 first with some detail, and just sketch Thm 4.2 (i.e., opposite from the current writing)? I suggest this because Laa16 style sieving is always worse than Thm 4.4, so it becomes just an optional part in our paper.}
% \yixin{The current writing has the advantage of starting from the
% easiest and building up on that no? I am not sure it would be very
% clear if we started with Thm 4.4.}
% \minki{I see, make sense. Regarding the limited time we have, it seems okay to preserve the current order.}

%\bg{Thanks for your suggestion! I'll apply that.}
\subsection{Quantum search inside the $\alpha$-close filters}
\label{section: grover_to_laa}
Applying a quantum search algorithm for searching close vectors in $B_i$, \cite{Laa16} obtains quantum sieving in time $2^{0.2653d+o(d)}$. 
% Following is the unoptimized/optimized time complexity of the quantum LSF given in \cite{Laa16}.

\begin{theorem}[{\cite[Section 14.2.10]{Laa16}}]
    There is a quantum lattice sieving algorithm, given a list of $n$ random input vectors and $\alpha,\beta \in (0,1)$, that outputs $\Omega(n)$ reduced vectors using the quantum LSF with $t=\mathcal{W}_{d}(\alpha, \beta, \pi/3)^{-1}$ filters. The running time of the algorithm is
    %and parameter $\alpha, \beta$ to output $O(n)$ reduced vectors with query complexity
    \begin{equation}
        T_1 = nt \cdot \mathcal{C}_{d}(\beta) + nt \cdot \mathcal{C}_{d}(\alpha)+n \cdot \left[ nt \cdot \mathcal{C}_{d}(\alpha)\cdot \mathcal{C}_{d}(\beta) \right]^{1/2}.
        \label{eq: laa_lsf_wrt_query_unopt}
    \end{equation}
    % where $n$ is defined in \eqref{eq: n}.
    In particular, the optimal choice gives
    the time complexity of $2^{0.2653d+o(d)}$.
    \label{thm: laa_lsf_terminal_result}
\end{theorem}
\begin{proof}
    The first two terms are derived in exactly the same way as in~\Cref{thm: classical_lsf}. 
    When finding the reducing pair for the query vector $\vecv$ in the union $D_{v}$ of the $B_i$
    such that $B_i$ is $\alpha$-close to $v$,
    quantum amplitude amplification is used with the oracle function 
    % $F_{\vecv}: \bigcup_{f\in \mathcal{F}_{\vecv, \alpha}}{B_f} \rightarrow \{0, 1\}$ defined as follows.
    \begin{equation}
        F_{\vecv}(\vecw) = 1 \text{  iff }\norm{\vecv - \vecw} \le R'
        \label{eq: laa_lsf_wrt_query_oracle_function}
    \end{equation}
    along with the reflection on the uniform superposition of the vectors in $D_v$.  Since the expected size of the search space is $\left( n\cdot\mathcal{C}_{d}(\alpha) \right) \cdot \left( t\cdot\mathcal{C}_{d}(\beta) \right)$, we obtain the desired time complexity. 
    
    % For each query vector $\vecv$, the expected number of $\beta$-close center vectors in filter buckets whose filters are $\alpha$-close to $\vecv$ is $\left( n\cdot\mathcal{C}_{d}(\alpha) \right) \cdot \left( t\cdot\mathcal{C}_{d}(\beta) \right)$, which represents the size of the Grover search space. Therefore, the third term in Equation~\ref{eq: laa_lsf_wrt_query_unopt} is induced by taking square root over it, followed by multiplying with the number of query vectors.
    Letting the three terms be equal to each other, we get $\alpha=\beta=\frac{\sqrt{3}}{4}=0.4330$ and $T_1=\left( \frac{13}{9} \right)^{d/2+o(d)}=2^{0.2653d+o(d)}$.
\end{proof}

% While there is no explicit remark on the required QRACM, 
To implement the reflection operator,
all of $S=nt \cdot \mathcal{C}_{d}(\alpha)\cdot \mathcal{C}_{d}(\beta)$ vectors in relevant $B_i$'s should be stored in QRAM. With the parameters used in the optimized setting, the $S=2^{0.1155d+o(d)}$ size QRAM is required.
The trade-off relation between the size of QRAM and time can be induced using~\Cref{lemma: limited_qracm_grover}. If the allowed QRAM is bounded, we can get the time complexity as follows. For convenience, we denote the size of QRAM as $\gamma^d$, instead of $2^s$ in~\Cref{lemma: limited_qracm_grover}.

%\jh{Until this point, we use the word ``QRAM``, but in below, we use ``QRACM``. Is it better to change all QRACMs to QRAMs?}\minki{See Chat in the overleaf. We never care about QRAQM so using simpler QRAM may be better}
% \yixin{There are no ``queries'' in the theorem below. Is that
% just the time complexity?} \bg{I understood that the query complexity would be applied to only third term. But should it be corrected to "time" complexity although the three terms are of same order?}
\begin{theorem}
    The time complexity of algorithm in~\Cref{thm: laa_lsf_terminal_result} with a QRAM of size $\gamma^d$ for $1\le \gamma\le 13/12$ is
    \begin{equation}
    % \begin{split}
        T_2\ =\  nt\cdot \mathcal{C}_{d}(\beta)+nt\cdot \mathcal{C}_{d}(\alpha)+n^2t\cdot \mathcal{C}_{d}(\alpha)\cdot \mathcal{C}_{d}(\beta)/\gamma^{d/2} %\\
        % & (1\le s^d \le nt\cdot \mathcal{C}_{d}(\alpha)\cdot \mathcal{C}_{d}(\beta))
        \label{eq: laa_lsf_tradeoff_wrt_query_unopt}
    % \end{split}
    \end{equation}
    For the optimal choice of $\alpha$
    and $\beta$, we have $T_2=\left( \frac{3\gamma}{3\gamma-1} \right)^{d/2+o(d)}.
    % \ (1\le s \le 13/12)
    $
    \label{thm: laa_lsf_tradeoff_wrt_query_result}
\end{theorem}
\begin{proof}
After finding filters $\alpha$-close to the query vector, one uses Grover search over the list of vectors in the filter buckets with the expected number of entries being $\left(n \mathcal{C}_{d}(\alpha) \right)\cdot\left( t \mathcal{C}_{d}(\beta) \right)$. The oracle function for the query vector $\vecv$ is $F_{\vecv}$ defined as~\cref{eq: laa_lsf_wrt_query_oracle_function}.
With~\Cref{lemma: limited_qracm_grover}, a QRAM of size 
$\gamma^d$ gives a $\gamma^{d/2}$ improvement in time, resulting in \cref{eq: laa_lsf_tradeoff_wrt_query_unopt}. The size $\gamma^d$ is bounded by the size of search space $nt\cdot \mathcal{C}_{d}(\alpha)\cdot \mathcal{C}_{d}(\beta)$, giving the bound of $\gamma.$

We can optimize the running time by equalizing the three terms.
This gives $\alpha=\beta=\sqrt{1-\frac{3\gamma}{4}}
% \ \left( 1\le s \le \frac{4}{3} \right)
$ and $T_2=\left( \frac{3\gamma}{3\gamma-1} \right)^{d/2}$
as stated. 
% \minki{I removed the discussion on the range of gamma because it's too long.}
% The range of $s$ comes from solving $s^d \le nt\mathcal{C}_{d}(\alpha)^2$, and it is easily checked that $s=1$ gives $T_1^{C}=\left( \frac{3}{2} \right)^{d/2+o(d)}=2^{0.2925d+o(d)}$ and $s=\frac{13}{12}$ gives $T_1^{Q}=\left( \frac{13}{9} \right)^{d/2+o(d)}=2^{0.2653d+o(d)}$ with $s^d=2^{0.1155d}$ QRAM.
\end{proof}

This algorithm
% in~\Cref{thm: laa_lsf_tradeoff_wrt_query_result} 
can be improved by searching the reducing pairs within each filter bucket rather than with respect to buckets corresponding to each query vector, as shown in~\cite{CL21}. 
Although its optimal time complexity remains the same, the required QRAM is reduced by halving the exponent of $\gamma^d$.

\begin{theorem}
    There is a quantum lattice sieving algorithm that solves the problem in~\Cref{thm: laa_lsf_tradeoff_wrt_query_result} with a QRAM of size $\gamma^d$ with time complexity:
    \begin{equation}
        T_3\ = \begin{cases}
            nt\cdot \mathcal{C}_{d}(\beta)+nt\cdot \mathcal{C}_{d}(\alpha)+ n^2t\cdot\mathcal{C}_{d}(\alpha)\cdot \mathcal{C}_{d}(\beta)/\gamma^d\\
            \ \ \ \ \ \text{if }1 \le \gamma^d \le \big[ nt \cdot \mathcal{C}_{d}(\alpha)\cdot \mathcal{C}_{d}(\beta) \big]^{1/2}\\
            nt\cdot \mathcal{C}_{d}(\beta)+nt\cdot \mathcal{C}_{d}(\alpha)+ n\cdot \left[ nt \cdot \mathcal{C}_{d}(\alpha)\cdot \mathcal{C}_{d}(\beta)\right]^{1/2}\\
            \ \ \ \ \ \text{if }\gamma^d > \big[ nt \cdot \mathcal{C}_{d}(\alpha)\cdot \mathcal{C}_{d}(\beta) \big]^{1/2}
        \end{cases}
        \label{eq: laa_lsf_tradeoff_wrt_filter_unopt}
    \end{equation}
    For an optimal choice of $\alpha$ and $\beta$, the time complexity is $T_3=\left( \frac{3\gamma^2}{3\gamma^2-1} \right)^{d/2}$ for $1\le \gamma \le \sqrt{\frac{13}{12}}$, and $T_3=\left( \frac{13}{9} \right)^{d/2+o(d)}=2^{0.2653d+o(d)}$ for $\gamma > \sqrt{\frac{13}{12}}$.
    \label{thm: laa_lsf_tradeoff_wrt_filter_result}
\end{theorem}
\begin{proof}
    In contrast to previous theorems, filter buckets are filled with $\beta$-close center vectors and $\alpha$-close query vectors in $nt\cdot \mathcal{C}_{d}(\beta)+nt\cdot \mathcal{C}_{d}(\alpha)$ time. We separate each with $B_{\mathcal{F}, \beta}$ and $B_{\mathcal{F}, \alpha}$ respectively, and each bucket stores $n\mathcal{C}_{d}(\beta)$ center vectors and $n\mathcal{C}_{d}(\alpha)$ query vectors respectively. The oracle function $F_{\mathcal{F}}$ for Grover search is defined as follows.
    \begin{equation}
        F_{\mathcal{F}}(\vecv, \vecw)=1 \text{ iff }\norm{\vecv - \vecw} \le R'
        \label{eq: laa_lsf_wrt_filter_oracle_function}
    \end{equation}
    As explained in the proof of~\Cref{thm: classical_lsf},
    the algorithm will output an expected $\Omega(n)$ reduced vectors,
    and therefore we expect $\Omega(n/t)$ reducing pairs in each filter on average, by the uniform randomness of vectors and filters.
    Using~\Cref{lemma: limited_qracm_grover_pair}, the number of queries to find all $\Omega(n/t)$ solutions in each filter is
    \begin{equation}
        T_{\text{per filter}} = \begin{cases}
            n\mathcal{C}_{d}(\alpha)  \cdot n\mathcal{C}_{d}(\beta)  / \gamma^d\\
            \ \ \ \ \ \text{if } 1 \le \gamma^d < \big[ n\mathcal{C}_{d}(\alpha) \cdot n\mathcal{C}_{d}(\beta) / (n/t) \big]^{1/2},\\
            \big[ n\mathcal{C}_{d}(\alpha)\cdot  n\mathcal{C}_{d}(\beta) \cdot (n/t) \big]^{1/2} \\
            \ \ \ \ \ \text{if } \big[ \left( n\mathcal{C}_{d}(\alpha) \right) \cdot \left( n\mathcal{C}_{d}(\beta) \right) / (n/t) \big]^{1/2} < \gamma^d.
        \end{cases}
        \label{eq: laa_lsf_tradeoff_Wrt_filter_perfilter}
    \end{equation}
    Multiplying \cref{eq: laa_lsf_tradeoff_Wrt_filter_perfilter} 
    by $t$ gives the third term of $T_3$ for both cases.
    % We can optimize the running time by equating all three terms,
    % and the result is independent of $\gamma$ in case 2. 
    Optimizing the three terms, we obtain the result.
    The range of $\gamma$ comes from the condition of $\gamma^d$ separating the two cases.
\end{proof}

In conclusion, using Grover search to buckets with respect to each filter reduces the QRAM size, achieving $T_3=2^{0.2653d+o(d)}$ with only $2^{0.05778d+o(d)}$ QRAM. It can be easily checked that the result given in \cite[Fig. 5]{CL21} coincides with the relation in~\Cref{thm: laa_lsf_tradeoff_wrt_filter_result}.

%\jh{I really appreciate if you check that the title of Section 4.3 is acceptable}
%\yixin{Suggested title: Quantum search over the $\alpha$-close filters}

\subsection{Quantum search over the $\alpha$-close filters}
\label{section: grover_to_heiser}
\cite{Heiser21} further improves the time complexity of quantum LSF by using Grover search to find a candidate bucket that contains a
vector forming a reducing pair with the vector currently being
processed. This gives a quadratic speedup on the second term of~\Cref{thm: classical_lsf}.

\begin{lemma}[{\cite[Section~4]{Heiser21}}]\label{lem:heiser_sampler}
    Let $\mathcal{F}_1,\ldots,\mathcal{F}_t$ be a list of LSF filters 
    constructed via random product codes, and $\alpha\in(0,1)$.
    There is a randomized (classical) algorithm that given
    a vector $\vec{w}$, returns a uniformly random
    (pseudo\footnote{See proof, a pseudo $\alpha$-close vector is always $(\alpha-\varepsilon)$-epsilon
    for some $\varepsilon=O(\tfrac{1}{d})$.})
    $\alpha$-close filter to $\vec{w}$ in time $\poly(d)$.
    This sampler requires $2^{o(d)}$ preprocessing time.
    The algorithm only requires
    $O(\log(d)\cdot \log(S(\vec{w}))$ random coins to
    output a filter on any given input $\vec{w}$, where
    $S(\vec{w})$ is the number of $\alpha$-close filters to $\vec{w}$.
\end{lemma}
\begin{proof}
    Almost everything is a direct consequence of \cite[Section~4.5]{Heiser21}.
    The only nontrivial part is estimating the number of random coins
    which is not done in the analysis of \cite{Heiser21}.
    The sampler works by constructing a tree in the pre-processing phase
    and pre-computing some values using dynamic programming. To sample
    a (pseudo) good $\alpha$-filter, the algorithm simply goes down
    one branch of the tree randomly. At each node, it sample random
    an integer $\ell$ between $1$ and ``$\mathcal{L}^R(x)$''
    which is the number of ``good leaves'' in the tree (after the
    $R$-discretization). Furthermore, \cite[Section~4.5.3]{Heiser21}
    shows that if a filter is pseudo $\alpha$-close then it is
    $(\alpha-\varepsilon)$-close where $\varepsilon=O(\tfrac{1}{d})$.
    Therefore, for large enough $d$, the number of pseudo $\alpha$-close
    filters is essentially the same as the number of $\alpha$-close filters.
    It follows that the algorithm only needs $\log_2(S)$
    random coins at each level of the tree where $S$ is the number of
    $\alpha$-close filters. Finally, there are only $O(\log d)$
    levels in the tree.
\end{proof}

\begin{theorem}[{\cite{Heiser21}}]
    There is a quantum lattice sieving algorithm that given a list $L$ of $O(n)$ random input vectors, uses quantum LSF with $t=\mathcal{W}_{d}(\alpha, \beta, \pi/3)^{-1}$ filters and parameters $\alpha, \beta$ to output $\Omega(n)$ reduced vectors
    in time
    \begin{equation}
        T_{4} = nt\cdot \mathcal{C}_{d}(\beta) + n\left[t \cdot \mathcal{C}_{d}(\alpha) + nt\cdot \mathcal{C}_{d}(\alpha)\cdot \mathcal{C}_{d}(\beta) \right]^{1/2}.
        \label{eq: heiser_lsf_unopt}
    \end{equation}
    In particular, for an optimal choice of $\alpha$ and $\beta$, the complexity is $2^{0.2571d+o(d)}$.
    \label{thm: heiser_lsf_terminal_result}
\end{theorem}
\begin{proof}
    In~\Cref{thm: laa_lsf_terminal_result,thm: laa_lsf_tradeoff_wrt_query_result,thm: laa_lsf_tradeoff_wrt_filter_result}, 
    filters which are $\alpha$-close to a query vector are found
    classically in time $O\left( nt\cdot \mathcal{C}_{d}(\alpha) \right)$ with the help of the random product code (RPC). \cite{Heiser21}
    explains how to obtain a sampler that can return a random
    (pseudo) $\alpha$-close filter in $\poly(d)$ time. This sampler
    only requires $2^{o(d)}$ preprocessing time.
    % discretizes the good leaf ($\alpha$-close) condition in the tree and does preprocessing in non-dominant $2^{o(d)}$ time. After the preprocessing, one can sample a (pseudo) $\alpha$-close filter in $\tilde{O}(1)$ time. 
    By turning this sampler into a quantum circuit,
    we can get a superposition of $\alpha$-close filters with respect to the query vector. Then the QRAM returns the $\beta$-close center vectors $\vecw$ for each filter $f$. Hence, for the query vector $\vecv$, we get the state 
    \begin{equation}
        \sum_{\substack{i \\ \mathcal{F}_{i, \alpha} \text{is $\alpha$-close to $\vecv$}}}{\sum_{\vecw \in B_i}{|\mathcal{F}_{i, \alpha}\rangle|\vecw\rangle}}.
    \end{equation}
    % \begin{equation}
    %     \sum_{\substack{i \\ \mathcal{F}_{i, \alpha} \text{is $\alpha$-close to $\vecv$}}}{|i\rangle|0\rangle} \xrightarrow{\text{QRAM}} \sum_{\substack{i \\ \mathcal{F}_{i, \alpha} \text{is $\alpha$-close to $\vecv$}}}{\sum_{\vecw \in B_i}{|i\rangle|\vecw\rangle}}.
    %     \label{eq: heiser_qracm_behavior}
    % \end{equation}
    
    We can then apply Grover's algorithm with the oracle function $F_{\vecv}$ defined by
    \begin{equation} \begin{split}
        F_\vecv(\mathcal{F}_{i, \alpha}, \vecw) = 1 & \text{ iff } \norm{\vecv - \vecw} \le R'
        \label{eq: heiser_lsf_oracle_function}
    \end{split} \end{equation}
    to obtain a quadratic speedup for the second term as well.

    Optimization by letting the constituting terms equal to each other gives $\alpha=0.4434, \beta=0.5$, thereby the complexity is $2^{0.2571d+o(d)}$.
\end{proof}

As far as we know, $2^{0.2571d+o(d)}$ is the best complexity achievable by using only a QRAM. We can now apply~\Cref{lemma: limited_qracm_grover} to~\Cref{thm: heiser_lsf_terminal_result} in order to get a time-QRAM
trade-off relation.

\begin{theorem}
    The time complexity of algorithm in~\Cref{thm: heiser_lsf_terminal_result} with a QRAM of size $\gamma^d$ is 
    \begin{equation}
        T_5 = nt\cdot \mathcal{C}_{d}(\beta)+\left[ nt\cdot \mathcal{C}_{d}(\alpha) + n^2t\cdot \mathcal{C}_{d}(\alpha) \cdot \mathcal{C}_{d}(\beta) \right]/\gamma^{d/2}.
        \label{eq: heiser_lsf_tradeoff_unopt}
    \end{equation}
    \label{thm: heiser_lsf_tradeoff_result}
    In particular, for an optimal choice of $\alpha$ and $\beta$, the time complexity is $T_5 = \left(\gamma-\frac{2}{3}+\frac{2}{3}\sqrt{1-\frac{3}{4}\gamma} \right)^{-d/2}$ for $1 \le \gamma \le 1.07122$.
\end{theorem}
\begin{proof}
    For each query $\vecv$, oracle function is defined as \cref{eq: heiser_lsf_oracle_function}. Therefore,~\Cref{lemma: limited_qracm_grover} can be used for bounded QRAM, with the search space of $nt\cdot \mathcal{C}_{d}(\alpha)+n^2t\cdot \mathcal{C}_{d}(\alpha)\cdot \mathcal{C}_{d}(\beta)$. Also, $\gamma$ should satisfy $\gamma^d \le t\cdot \mathcal{C}_{d}(\alpha)+nt\cdot \mathcal{C}_{d}(\alpha)\cdot \mathcal{C}_{d}(\beta)$.
    Letting three terms equal to each other, the values of parameters are $\alpha=\sqrt{1-\frac{3}{4}\gamma}$ and $\beta=0.5$, giving $T_5 = \left(\gamma-\frac{2}{3}+\frac{2}{3}\sqrt{1-\frac{3}{4}\gamma} \right)^{-d/2}$. These parameters also determine the range of $\gamma$ by the constraint given above in this proof.
\end{proof}
     
    % In this section, we apply the result of Grover search with limited QRACM in Lemma~\ref{lemma: limited_qracm_grover}, \ref{lemma: limited_qracm_grover_pair} to algorithms in \cite{Laa16} and \cite{Heiser21}.
    % Although the simplified assumption in writing data to QRACM is still valid for these results, it gives nice interpolations between the result of those works and the classical one in \cite{BDGL16}. Also, the complexity is purely a cost of queries, as a cost of read/write operations are ignored.

In summary, by using a version of Grover search with bounded QRAM to the quantum sieving algorithm, we get an interpolation between classical LSF in \cite{BDGL16} and the quantum LSF in \cite{Laa16,Heiser21}. Those results are summarized in~\Cref{fig: quantum_lsf_sieving_result}, showing the time complexity as a function of allowed QRAM size.

\begin{figure}
    \centering
    \begin{tikzpicture}[scale=0.75, every node/.style={transform shape}]
        \begin{axis}[
            xlabel={$\tfrac{1}{d}\log_2(\text{QRAM size})$},
            xtick distance=0.05,
            minor x tick num=10,
            ylabel={$\tfrac{1}{d}\log_2(\text{complexity})$},
            ytick distance=0.01,
            height=5cm,
            width=10cm  %jh) height & width added
        ]
            \addlegendentry{$T_2$ (\Cref{thm: laa_lsf_tradeoff_wrt_query_result})}
            \addplot[smooth,blue] table {laa_lsf_wrt_query.csv};
            \addlegendentry{$T_3$ (\Cref{thm: laa_lsf_tradeoff_wrt_filter_result})}
            \addplot[smooth,red] table {laa_lsf_wrt_filter.csv};
            \addlegendentry{$T_5$ (\Cref{thm: heiser_lsf_tradeoff_result})}
            \addplot[smooth,green] table {heiser_lsf.csv};
        \end{axis}
    \end{tikzpicture}
    \caption{Trade-off relations given in~\Cref{thm: laa_lsf_tradeoff_wrt_query_result,thm: laa_lsf_tradeoff_wrt_filter_result,thm: heiser_lsf_tradeoff_result}.
    % Theorems~\ref{thm: laa_lsf_tradeoff_wrt_query_result},~\ref{thm: laa_lsf_tradeoff_wrt_filter_result}~and~\ref{thm: heiser_lsf_tradeoff_result}. 
    The top-left point represents the result of classical LSF \cite{BDGL16}. The bottom-right point
    of each line represents the result of quantum LSF with no constraint on the QRAM size. In particular, we recover the results of \Cref{thm: laa_lsf_terminal_result} \cite{Laa16,CL21} (blue and red), and obtain a new trade-off from \Cref{thm: heiser_lsf_terminal_result} \cite{Heiser21} (green).}
    \label{fig: quantum_lsf_sieving_result}
\end{figure}
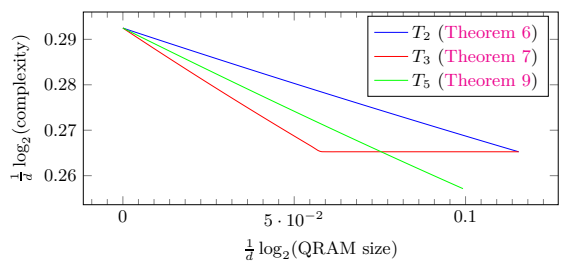

\subsection{Discussion on the QRAM model}
\label{sec:discussion_QRAM}

Recall that the trade-off obtained in~\Cref{thm: heiser_lsf_tradeoff_result} is under the QRAM model
of~\Cref{sec:qram}. In particular, we assumed that writing the data
to, or overwriting (i.e. reset and load another data) the QRAM can be done in time $O(1)$. This assumption is clearly unrealistic but 
simplifies the proof to focus on the time-QRAM trade-off.
Here, we present an alternative point of view to justify this
assumption. In this alternative model, we replace the reload/overwrite
operation above by two operations:
\begin{itemize}
    \item Create: this operation loads $N$ classical data into
        a QRAM of size $N$ and takes time $O(N)$. Once created,
        the QRAM can be used as many times as wanted but the content
        cannot be changed anymore.
    \item Connect: this operation takes an already created QRAM
        and connects it to a quantum circuit so it can be used.
        This operation takes time $O(1)$.
\end{itemize}
The motivation behind this model is to view a QRAM as a physical object
(say a quantum chip) where we somehow hardcode the data.
Creating such an object is surely expensive and we assume takes time
linear in the number of data stored. Once the object is created,
it can be physically stored somewhere and, only when needed, connected
to a quantum circuit to make I/O accesses. Connecting it to a circuit
intuitively should not depend on the size of the QRAM. Importantly,
in this model, it is entirely possible to create \emph{many} distinct
QRAMs.
\medskip

Let us now see how this model applies to~\Cref{thm: heiser_lsf_tradeoff_result}. In this algorithm, we need 
to store one list per filter (there are $t$ filters). Each list
contains $n\cdot \mathcal{C}_{d}(\beta)$ vectors on average. If
we only limit ourselves to QRAMs of size $\gamma^d$ then we need to create
$nt\cdot \mathcal{C}_{d}(\beta)/\gamma^d$ QRAM in total. The cost of
creating those QRAM is $n\cdot \mathcal{C}_{d}(\beta)$ which 
is equal to the first term of the complexity in~\Cref{thm: heiser_lsf_tradeoff_result} so in a certain sense, our complexity
almost includes the cost of creating those QRAM. Note, however,
that this will require to create and store an \emph{exponential} number
of QRAM. If we think about QRAM as physical objects, this might be
difficult but not impossible. Later in the algorithm, we only ever
use the QRAM initially created without any modifications. This means
that we can replace the ``reload'' operations by a ``connect'' that
runs in $O(1)$ and therefore obtain the same complexity.
\medskip

In summary, this alternative model shows that our model is not
as unrealistic as it appears if we think about QRAMs as physical
objects. It of course still relies on the assumption that a QRAM
can be implemented efficiently but this out of the scope of this
paper.

\section{Application of QRAM Trade-off to Symmetric Key}
\label{section: app_to_symmetric_key}
In~\Cref{section: quantum_lsf_sieving}, we introduced lattice sieving with a bounded QRAM based on~\Cref{lemma: limited_qracm_grover,lemma: limited_qracm_grover_pair},
% to solve lattice sieving, 
which is the basis of public key cryptography. In this section, we show how to apply these methods to symmetric-key-related problems.
% other problems, which may have potential implications for cryptanalysis in symmetric key cryptographic schemes.

\paragraph{Quantum collision finding.}
The collision finding problem (CF) asks to find two elements $x_1, x_2 \in X$ such that $f(x_1)=f(x_2), x_1 \ne x_2$ given a (random) function $f: X \rightarrow Y$. It is known that $\Omega(N^{1/2})$, for $|X|=N$,
is a lower bound on the number of queries to $f$.
With Grover search, \cite{BHT98} shows that $O(N^{1/3})$ queries are sufficient given access to a QRAM of size $O(N^{1/3})$.
\cite{CNS17} gives an algorithm to solve the same problem with $O(N^{2/5})$ queries, but without using any QRAM.
The main idea of quantum algorithms for CF is to precompute the values of $f$ on some $x$'s first, and then apply Grover or amplitude amplification for the rest of the data.
While the algorithm in \cite{BHT98} stores data in a large QRAM, \cite{CNS17} stores them in a classical memory, and implements a
search oracle with quantum gates. In other words, \cite{CNS17} converts the amount of QRAM into the time required to compare data one by one with gate operations. The time complexity of the algorithm in \cite{CNS17} is in \cref{eq: cns_terminal_time_unopt} below.

\begin{equation}
    T = 2^{l+r/2} + 2^{(n-r-l)/2}\left( 2^{r/2} + 2^{l} \right)
    \label{eq: cns_terminal_time_unopt}    
\end{equation}
\begin{equation*}
    \left( |X|=2^n,\ 2^l \text{ : number of pre-calculated data, }r \in \left[0, n\right] \text{ : parameter} \right)
\end{equation*}

Optimizing the parameters gives $r=\frac{2n}{5}$ and $l=\frac{n}{5}$, 
resulting in $T=2^{2n/5}$. The required classical memory is $2^l=2^{n/5}$, which is also the cost of precalculating the data.

If a QRAM of size $2^\gamma$ is allowed, we can use the idea in~\Cref{lemma: limited_qracm_grover} to obtain a speed up. We divide
the pre-calculated data into blocks of size $2^s$, load them into 
the QRAM, and access it as a membership query.
This results in an algorithm of time complexity:

\begin{equation}
    T' = 2^{l+r/2} + 2^{(n-r-l)/2}\left( 2^{r/2} + 2^{l-\gamma} \right), \ \ \ (0 \le \gamma \le l).
    \label{eq: cns_tradeoff_time_unopt}
\end{equation}

Optimizing the parameters gives $T'=2^{2n/5-\gamma/5}$ with $2^{n/5+2\gamma/5}$ classical memory and a QRAM of size $2^{\gamma}$.
While this method has fewer effects on $T$ than the outer parallelization introduced in \cite[Section 5.3]{CNS17}, it is still meaningful in that the approach is orthogonal to the outer parallelization; thereby, further improvement is possible by combining them.
%with some QRAM allowed.
Also, outer parallelization needs $2^\gamma\times$ more qubits to achieve $2^{2n/5-3\gamma/5}$ complexity, while the above $2^{2n/5-\gamma/5}$ complexity with $2^\gamma$ QRAM causes only $O(n+\gamma)$ additional qubits.

%\jh{Need to make subsection for Q collision finding? - If the collision finding is the only one given in this section, then I should modify the paragraph above, as well as the corresponding paragraph in the introduction.}
%\minki{Can we do something similar for the multi-target preimage search?}
%\jh{I write a new paragraph, applying the idea to the multi-target preimage search. I think the calculation is quite complicated due to the separation of range of $t, s$. }
%\minki{I think, if you are confident on your calculation, we can just summarize the result without full description because the same technique used multiple times in this paper already.}
%\jh{I double-checked my result, and seems correct as it coincides with the result in \cite{CNS17} with $\gamma=0$. However, there is no reference for the full-QRAM case $\gamma=n/3$, which gives $T=2^{n/3}$ according to my calculation. \cite{CNS17} says that ``naive`` Grover gives $T=2^{n/2}$, but it doesn't seem to be optimal. Moreover, if QRAM is allowed in MTPS problem, the first term can also be modified, not $2^t$. I think more works may give more results, which look almost similar to the result of Section 4. - Please erase this comment if it's enough}
% \minki{The reason of $2^{n/2}$ because in their paper the assume that they consider memory * time complexity. Assuming efficient QRAM gives a better result like $2^{(n-t)/2}$.}

\paragraph{Multi-target preimage search.}
Given a (random) permutation $H : X \rightarrow X$ and a set $T=\{y_1, \cdots, y_{2^t}\}\subset X$, the multi-target preimage search problem (MTPS) asks to find $i\in \{1, \cdots, 2^t\}$ and $x\in X$ such that $H(x)=y_i$.
\cite{CNS17} also gives a quantum algorithm to solve the MTPS problem without using any QRAM. The time complexity is 
\begin{equation}
    T = 2^{t} + 2^{(n-t)/2}\left(2^{r/2} + 2^{t-r} \right),
    \qquad |X|=2^n, r\in[0, n]\text{ : parameter}
    \label{eq: cns_preimage_terminal_time_unopt}
\end{equation}
If we can freely choose all parameters, the optimization gives $t=\frac{3n}{7}$ and $r=\frac{2t}{3}$. However, $t$ is usually a given parameter (i.e., the number of target images) in the MTPS problem, and the time complexity is expressed as a function $t$. If $t\ge \frac{3n}{7}$, then we can ignore some data to achieve the complexity of $2^{3n/7}$, while $2^{n/2-t/6}$ is the optimal complexity of~\cite{CNS17} when $t < \frac{3n}{7}$.

With a QRAM of size $2^\gamma$, the modified time complexity becomes
\begin{equation}
    T' = 2^{t} + 2^{(n-t)/2}\left(2^{r/2} + 2^{t-r-\gamma} \right), \ \ \ (0 \le \gamma \le t-r).
    \label{eq: cns_preimage_tradeoff_time_unopt}
\end{equation}
Optimizing the parameters gives $T'=2^{3n/7-2\gamma/7}$ for $t\ge \frac{3}{7}n-\frac{2}{7}\gamma$, and $T'=2^{n/2-t/6-\gamma/3}$ for $1\le t<\frac{3}{7}n-\frac{2}{7}\gamma$. In both cases, the size of the QRAM is bounded by $2^\gamma \le 2^{n/3}$. 

In summary, our strategy can be applied to the problems related to symmetric-key cryptography and can be applied to collision attacks on operation modes as discussed in~\cite{CNS17}.
% \\
% \jh{I really appreciate if check the below sentence}
% In summary, our model can also be applied to problems which have relations with symmetric-key, and take advantage while searching the solution space of the function.
\section{Lower Bounds with Bounded QRAM}\label{sec:lower}
This section establishes the (conditional) lower bounds for the hash-based nearest-neighbor algorithms and lattice sieving in the bounded QRAM setting. Some parts of this section are adapted from~\cite{KL21}.
% We first introduce the near-neighbor problem in $\Sd$ and the model of hash-based near-neighbor search algorithms adapting~\cite{KL21}.
% \minki{We need to clarify a little about the validity of the model, especially regarding the reusable quantum walks~\cite{BCSS23}.}

We note that our lower bound is actually about the cryptographically relevant near-neighbor problems. To our knowledge, almost all sieving variants, such as the closest vector problem with preprocessing~\cite{DLW20,Laa20}, use the near-neighbor subroutines; thereby, our lower bound applies. We refer to the discussion in~\cite{KL21} for the implication of the near-neighbor lower bounds.

\subsection{The problems, models, and classical lower bounds}

\subsubsection{The near-neighbor problem.}
The lattice sieving algorithms usually maintain a list of lattice vectors, and find the close pairs in the list to construct a list of shorter vectors. Formally, the following problem is solved as a subroutine.
% \minki{We need to add some words about almost all".}
\begin{definition}[Near-neighbor problem in $\Sd$]
    Let $\theta \in [0,\pi].$
    Let $L\subset \Sd$ be a finite subset of $\Sd$ whose elements are sampled uniformly at random from $\Sd$. In the near-neighbor problem, we 1) preprocess $L$ in a certain data structure, and 2) later, when a uniformly random $\vecx \in \Sd$ is queried, we find almost all\footnote{For example, we may ask to find 90\% of the vectors as in~\cite[Definition 4]{KL21}.} vectors $\vecy \in L$ such that $\inner{\vecx,\vec y} \ge \cos \theta$.
\end{definition}

\subsubsection{The model for classical hash-based algorithms.}
In the high-dimensional case, the hash-based approaches have been known to be the most effective. Roughly, this approach divides the space into smaller regions using multiple random hash functions. In each of these \emph{hash regions}, the algorithm searches for the close pairs, which is much more efficient than searching in the entire list.
Formally, the hash-based near-neighbor search algorithms can be described as follows.
\begin{definition}[Hash-based near-neighbor algorithms]
    Given the near-neighbor problem parameterized by $L\subset \Sd$ and $\theta\in[0,\pi]$, the hash-based near-neighbor algorithm preprocesses the list $L$ and processes queries $\vecx$ as given in~\Cref{alg:hashNNS}.
\end{definition}

\begin{algorithm}
\caption{The model of hash-based near-neighbor algorithms}\label{alg:hashNNS}
\begin{algorithmic}[1]
% \Require \textsc{Input:} The list $L$ and the target angle $\theta$
\Require \textsc{Scheme Parameters:}
\begin{itemize}
    \item[\sbt] $t\in\N$ \Comment{The number of hash regions}
    \item[\sbt] $U_1,...,U_t \subset \Sd$ \Comment{Hash regions for insertions}
    \item[\sbt] $Q_1,...,Q_t \subset \Sd$\Comment{Hash regions for queries}
    \item[\sbt] $\texttt{method} \in \{\texttt{Query},\texttt{FAS}\}$\Comment{Choices for the finalization}
\end{itemize}
\medskip

\Procedure{\textsc{Insert}}{$\vecy$}
    \ForAll{$i \in [t]$\textbf{ such that }$\vecy \in U_i$}\label{item-alg:hashNNS-insert-vecy}
        \State $B_i \gets B_i \cup \{\vecy\}$
    \EndFor
\EndProcedure
\smallskip

\Procedure{\textsc{Preprocess}}{$L$}
    \State $B_1,...,B_t \gets \emptyset$
    \ForAll{$\vecy \in L$}
        \State \textsc{Insert}$(\vecy)$
    \EndFor
\EndProcedure
\smallskip

\Procedure{\textsc{Query}}{$L,\theta$}
    \State $P\gets \emptyset$
    \ForAll{$\vecx \in L$}\Comment{For each $\vecx\in L,$ find near neighbors $\vecy$.}
        \State $C\gets \emptyset$
        \ForAll{$i \in [t]$\textbf{ such that }$\vecx \in Q_i$}\label{item-alg:hashNNS-query-vecx}
            \ForAll{$\vecy \in B_i$\textbf{ such that }$\inner{\vecx,\vecy}\ge \cos \theta$}\label{item-alg:hashNNS-query-close}
                \State $C \gets C \cup \{\vecy\}$
            \EndFor
        \EndFor
        \State $P \gets P \cup (\{\vecx\}\times C)$
    \EndFor
    % \State $C\gets \emptyset$
    % \ForAll{$i \in [t]$\textbf{ such that }$\vecx \in Q_i$}\label{item-alg:hashNNS-query-vecx}
    %     \ForAll{$\vecy \in B_i$\textbf{ such that }$\inner{\vecx,\vecy}\ge \cos \theta$}\label{item-alg:hashNNS-query-close}
    %         \State $C \gets C \cup \{\vecy\}$
    %     \EndFor
    % \EndFor
    \State \textbf{return }$P$
\EndProcedure
\smallskip

\Procedure{\textsc{FindAllSolutions}}{$L,\theta$}
    \State $P\gets \emptyset$
    \ForAll{$i \in [t]$}\Comment{For each $i \in [t],$ find close pairs $(\vecx,\vecy) \in A_i \times B_i$.}
        \State $C\gets \emptyset$
        \ForAll{$(\vecx,\vecy)\in A_i\times  B_i$\textbf{ such that }$\inner{\vecx,\vecy}\ge \cos \theta$}\label{item-alg:hashNNS-findallsol-close}
            \State $C \gets C \cup \{(\vecx,\vecy)\}$
        \EndFor
        \State $P\gets P\cup C$
    \EndFor
    % \State $C\gets \emptyset$
    % \ForAll{$\vecx,\vecy \in B$\textbf{ such that }$\inner{\vecx,\vecy}\ge \cos \theta$}\label{item-alg:hashNNS-findallsol-close}
    %     \State $C \gets C \cup \{(\vecx,\vecy)\}$
    % \EndFor
    \State \textbf{return }$P$
\EndProcedure
\smallskip

\Procedure{\textsc{Main}}{$L,\theta$}
    \State \textsc{Preprocess}$(L)$
    % \State $P \gets \emptyset$
    \If{$\texttt{method}=\texttt{Query}$}\Comment{used in \cite{KL21,Heiser21}}
        \State $P\gets \textsc{Query}(L,\theta)$
        % \ForAll{$\vecx \in L$}
        %     \State $P\gets P\cup (\{\vecx\}\times \textsc{Query}(\vecx))$
        % \EndFor
    \ElsIf{$\texttt{method}=\texttt{FAS}$}\Comment{used in \cite{CL21,BCSS23}}
        \State $A_1,...,A_t \gets \emptyset$\label{item:prep1}
        \ForAll{$\vecx \in L$}\Comment{Preprocessing $L$ regarding $Q_i$}\label{item:prep2}
            \ForAll{$i\in[t]$\textbf{ such that }$\vecx \in Q_i$}\label{item:prep3}
                \State $A_i \gets A_i \cup \{\vecx\}$\label{item:prep4}
            \EndFor
        \EndFor
        \State $P\gets \textsc{FindAllSolutions}(L,\theta)$
        % \ForAll{$i \in [t]$}
        %     \State $P\gets P\cup \textsc{FindAllSolutions}(B_i)$
        % \EndFor
    \EndIf
    \State \textbf{return }$P$
\EndProcedure
\end{algorithmic}
\end{algorithm}

Our model has two methods \texttt{Query} and \texttt{FAS} for \textsc{Main} function. 
The first method \texttt{Query}, which was originally used in~\cite{KL21}, searches for nearby vectors for each input vector using the function \textsc{Query}. On the other hand, the method \texttt{FAS} searches for pairs of close vectors in each hash region, reflecting the framework suggested in~\cite[Algorithm 1]{CL21}. 
% Considering a more general model, including both methods, is essential for the quantum lower bound proof. 
% \jh{I think this sentence has a grammar error, but don't know how to modify.}
The difference between the two methods does not affect the lower bounds and proofs as described below.

Kirshanova and Laarhoven~\cite[Theorem 2 and 3]{KL21}\footnote{Strictly speaking, they showed the result below for \texttt{Query} method. Still, the collision probability parts (Theorem 2) are irrelevant to the method, and the equal choices for the caps (Theorem 3) are argued by looking at the overall complexity, which does not depend on the choice of method, as shown below.} proved that choosing spherical caps of the same size for the hash regions gives the optimal algorithm. Following this, we assume that the hash regions are of the following form:
\begin{itemize}
    \item Choose $\alpha,\beta\in[-1,1]$ and draw  $\vecv_i\gets\Sd$  randomly for $i\in [t]$, and define 
    \begin{align}\label{eqn: optchoice}
        Q_i:= \{\vecz\in \Sd : \inner{\vecz,\vecv_i} \ge \alpha\}\text{ and }U_i:= \{\vecz\in \Sd : \inner{\vecz,\vecv_i} \ge \beta\}.
    \end{align}
\end{itemize}

% We recall the following theorem, which states that the spherical caps of the same size are the optimal choices for hash regions.
% \begin{theorem}[{\cite[Theorem 2,3, and 4]{KL21}}]
%     The following choices of hash regions are optimal among the hash-based methods described by~\Cref{alg:hashNNS}:
%     \begin{itemize}
%         \item Choose some $\alpha,\beta\in[-1,1]$ and draw random $\vecv_i\gets\Sd$ for $i\in [t]$, and 
%         \item define $Q_i:= \{\vecz\in \Sd : \inner{\vecz,\vecv_i} \ge \alpha\}$ and $U_i:= \{\vecz\in \Sd : \inner{\vecz,\vecv_i} \ge \beta\}$.
%     \end{itemize}
% \end{theorem}
% \minki{Maybe Theorem (?) 4 has a little dependency on the abstraction of the algorithm; so let me add some description about it later.}

\subsubsection{Query complexity.}
Most parts of computational cost in~\Cref{alg:hashNNS} stem from the operations related to the input vectors in $L$ and the filters.
Here, we prove that the time complexity lower bound in~\cite{KL21} can be extended to a query lower bound regarding the operations for input vectors and filters, where the explicit definition of the query is as follows. A more formal definition requires the black-box model for the vectors and filters, which can be found in~\Cref{app: lowerbound}.

\begin{costmodel}[Classical query complexity]\label{costmodel: cost}
    All operations related to the input vectors in $L$ and filters are done through oracle access. Among these operations, the costs of the following queries are counted as the query complexity.
    \begin{enumerate}
    % \item Inserting a vector $\vecx$ in some set $B_i$ (or $A_i$) at a unit cost. This is only used in the preprocessing.\label{item:cost_insert}
    \item Given a vector $\vecx$, sampling a random index $i\in[t]$ such that $\vecx \in Q_i$ can be done at a unit cost; the filter corresponding to $i$ is called by the \emph{relevant} filter. If the number of relevant filters is $|Z|$, finding all of them can be done at $1+|Z|$ unit costs. \footnote{It corresponds to the assumption in the bottom of~\cite[p.6]{KL21}, which reflects the advanced hash-based algorithms. The constant +1 is to address the case of $|Z|=0$.} This query is used in~\Cref{item-alg:hashNNS-insert-vecy,item-alg:hashNNS-query-vecx}.
    The algorithm can insert $\vecx$ in $A_i$ or $B_i$ corresponding to the relevant filter for free.\label{item:cost_enum}
    % \item Given an index $i\in [t]$, sampling a random vector in $B_i$ or $A_i$ can be done at a unit cost. If the number of such vectors is $|Y|$, finding all such vectors can be done at $1+|Y|$ unit costs.\label{item:cost_filterelt}
    \item Given two vectors $(\vecx,\vecy)$ as input, check if the inner product satisfies $\inner{\vecx,\vecy}\ge1/2$ or not at a unit cost. This query is used in~\Cref{item-alg:hashNNS-findallsol-close,item-alg:hashNNS-query-close}.\label{item:cost_inner}
    % \item Given a vector $\vecx$ and a character $C\in \{U,Q\}$ as input, output all indices denoting the regions $i \in[t]$ such that $\vecx \in C_i$. If the number of such indices $i$ is $|Z|$, it takes $1+|Z|$ unit costs\footnote{It corresponds to the assumption in the bottom of~\cite[p.6]{KL21}, which reflects the advanced hash-based algorithms.}\footnote{The constant +1 is to address the case of $|Z|=0$.}. This query is used in~\Cref{item-alg:hashNNS-insert-vecy,item-alg:hashNNS-query-vecx}.\label{item:cost_enum}
\end{enumerate}

% \begin{costmodel}[Classical query complexity]\label{costmodel: cost}
%     The vectors in the input list $L$ are given as handles. The following operations can be done via oracle access.
%     \begin{enumerate}
%     \item Given two vectors $(\vecx,\vecy)$ as input, compute the inner product $\inner{\vecx,\vecy}$ in unit cost. This query is used in~\Cref{item-alg:hashNNS-findallsol-close,item-alg:hashNNS-query-close}.\label{item:cost_inner}
%     \item Given a vector $\vecx$ and a character $C\in \{U,Q\}$ as input, output all indices denoting the regions $i \in[t]$ such that $\vecx \in C_i$. If the number of such indices $i$ is $|Z|$, it takes $1+|Z|$ unit costs\footnote{It corresponds to the assumption in the bottom of~\cite[p.6]{KL21}, which reflects the advanced hash-based algorithms.}\footnote{The constant +1 is to address the case of $|Z|=0$.}. This query is used in~\Cref{item-alg:hashNNS-insert-vecy,item-alg:hashNNS-query-vecx}.\label{item:cost_enum}
% \end{enumerate}
The summation of the unit cost incurred by the above oracle queries during the algorithm is called the query complexity of the algorithm.
\end{costmodel}

We can derive the following classical query lower bound by adapting~\cite{KL21} in this query model.
The formal statement (\Cref{lem: classical lower bound}) and proof of this theorem can be found in~\Cref{app: lowerbound} along with a finer formalization.

% , along with the formalization~\Cref{app: lowerbound}, can be found in and \Cref{lem: classical lower bound}. 
% \jh{Something is mixed. Lemma 8 is included in Appendix B, so does it mean "formalization~\Cref{app: lowerbound} can be found in~\Cref{lem: classical lower bound}}
% \yixin{I don't know how to fix that, @Minki can you check?}
\begin{theorem}[{\cite[Theorem 4, adapted]{KL21}}]\label{thm: classical lower bounds}
    The classical near-neighbor algorithm described by~\Cref{alg:hashNNS} has query complexity at least $2^{0.2925d+o(d)}$, regardless of the choice of the method.
\end{theorem}

\subsection{The quantum time-memory trade-off lower bounds}
Before proceeding to the quantum lower bound, we first introduce the model of quantum hash-based near-neighbor algorithms. We assume that the quantum algorithm also follows the framework given in~\Cref{alg:hashNNS} and implements post-processing (\textsc{Query} or \textsc{FindAllSolutions}) quantumly with the same purpose, while the \textsc{Preprocess} part remains classical. 
% Precisely, as noted in the comments of the algorithm, we assume that the purpose of \textsc{FindAllSolution} is to find close pairs $(\vecx,\vecy) \in A_i\times B_i$ for each $i\in [t]$ and one of \textsc{Query} is to find near neighbors $\vecy$ of each $\vecx \in L$, we do not strictly specify the details.

The known quantum algorithms indeed follow the same template with the quantum modifications for the procedure \textsc{Query} (e.g., in~\cite{Heiser21}) or \textsc{FindAllSolutions} (e.g., in~\cite{CL21,BCSS23}) in~\Cref{alg:hashNNS}.
% We mainly focus on the QRAM access in this model.

% \subsubsection{The use of QRAM.}
The quantum speedup is from the procedures \textsc{Query} and \textsc{FindAllSolutions}. In particular, the vectors in the relevant filters $B_i$, i.e., the vectors $\vecy \in L$ such that $\vecy \in U_i$ are stored in the QRAM and coherently accessed during these procedures. 
% The above-mentioned quantum lattice sieving algorithms use the quantum RAM to implement coherent access to the vectors in $B_i$. 
Compared to the classical model, reading the input vectors in the filters or the relevant filter indices can be done coherently through the QRAM access. The quantum cost model can be summarized as follows; again, see~\Cref{app: lowerbound} for a more detailed description.

\begin{costmodel}[Quantum query/QRAM complexity]\label{costmodel: quantum_cost}
    As in~\Cref{costmodel: cost}, the operations related to the input vectors and filters can be done via oracle access. Among these, the following queries are counted in the complexity.
    \begin{enumerate}
    % \item Inserting a vector $\vecx$ in some set $B_i$ (or $A_i$) at a unit cost. This is only used in the preprocessing.
    \item Given a vector $\vecx$, sampling a relevant filter can be done at a unit cost. Coherent access to some subset of the relevant filters can be done using the QRAM at a unit cost.\label{item:cost_qenum}
    % \item Given an index $i\in [t]$, sampling a random vector in $B_i$ (or $A_i$) can be done at a unit cost. Coherent access to some subset of such vectors can be done using the QRAM at a unit cost.\label{item:cost_qfilterelt}
    \item Given registers $\sum_{i,j} \alpha_{i,j}\ket{\vecx_i,\vecy_j,r}$ as input, compute the inner product \\
    $\sum_{i,j} \alpha_{i,j}\ket{\vecx_i,\vecy_j,r+1_{\inner{\vecx_i,\vecy_j}\ge 1/2}}$ in a unit cost where $1_{a\ge 1/2} =1$ if $a\ge 1/2$, otherwise 0.\label{item:cost_qinner}
    % \item Given two vectors in a computational basis, compute the inner product. More formally, it computes
    % \begin{align}
    %     \sum_{\vecx,\vecy,\in L, r \in \R} \alpha_{\vecx,\vecy,r} \ket{\vecx,\vecy,r} \mapsto
    %     \sum_{\vecx,\vecy,\in L, r \in \R} \alpha_{\vecx,\vecy,r} \ket{\vecx,\vecy,r + \inner{\vecx,\vecy}}
    % \end{align}
    % in unit cost. If the input is not of this form, it does nothing.
    % \label{item:cost_qinner}
    % \item Given a vector $\vecx$ and a character $C\in \{U,Q\}$ as input, output all indices denoting the regions $i \in[t]$ such that $\vecx \in C_i$. If the number of such indices $i$ is $|Z|$, it takes $1+|Z|$ unit cost.\label{item:cost_qenum}
\end{enumerate}
\end{costmodel}

Note that these queries require QRAM access.
% \jh{I think the sentence is outdated after modifying the items above. I think ``Among these queries,~\Cref{item:cost_qenum} uses the QRAM access.`` is correct one}
To formalize the usage of the QRAM, we introduce the following assumption on the coherent access to the input vectors, which roughly states that the coherent states of the input vectors are constructed only through the QRA(C)M access.
\begin{assumption}\label{assumption: QRAM}
    Let $V=\{\vecv_1,...,\vecv_k\}$ be an arbitrary subset of the input list $L$. The coherent quantum state
    \begin{align}\label{eqn: QRAM_list}
    % \sum_{i\in[k]} \alpha_i \ket{i}\ket{0} \mapsto 
    \sum_{i\in[k]} \alpha_i \ket{i}\ket{\vecv_i}
    \end{align}
    must be generated by the QRAM access to the list of $k$ vectors.
\end{assumption}
One way to interpret~\Cref{assumption: QRAM} is to say that there is no efficient way
to generate this superposition, except by using a QRAM. Indeed, it is clear that such a state can be generated
by a plain quantum circuit, but all currently known ways of doing so are inefficient in some way, e.g., requires large depth or qubits.\footnote{
One extreme is to sequentially read each vector, which requires a circuit of depth linear in the number
of elements but only a few qubits. The other extreme is to build a tree (like in classical  RAM) of
depth logarithmic in the number of elements but whose width (i.e., the number of qubits) is linear in the number
of elements. Some trade-offs between those two extremes are possible and suggest that one always needs either
a very deep circuit or a large number of qubits} Indeed, if we implement the above state in the circuit model, the number of gates must be $\Omega(k)$ \cite[Theorem V.2]{Jaq23}.

We also need a similar assumption on filter index access\footnote{In fact, the filter index itself can be coherently accessed without QRAM. However, whenever we want to check any information about the input vectors in the filter (e.g., checking if the filter is empty), it requires QRAM access.}. This operation is used, e.g., in~\Cref{item-alg:hashNNS-query-vecx} of the quantum version of~\Cref{alg:hashNNS}~\cite{Heiser21}.

\begin{assumption}\label{assumption: QRAM_filter}
    Let $F=\{f_1,...,f_k\} $ be an arbitrary subset of the set of the filter indices. The coherent quantum state
    \begin{align}\label{eqn: QRAM_filter}
    % \sum_{i\in[k]} \alpha_i \ket{i}\ket{0} \mapsto 
    \sum_{i\in[k]} \alpha_i \ket{i}\ket{f_i}
    \end{align}
    must be generated by QRAM access to the list of $k$ vectors.
\end{assumption}

These assumptions imply that when the algorithm uses a QRAM of size $2^s$, each register of the algorithm contains at most $2^s$ different vectors or filter indices. This is an important observation in the proof.

\subsubsection{Quantum lower bound.}
Our quantum lower bounds for the bounded QRAM algorithms are summarized as follows.
\begin{theorem}\label{thm:qsieve_lowerbound}
    The quantum near-neighbor algorithm described by~\Cref{alg:hashNNS} with
    a QRAM of size $2^s$ has query complexity at least $2^{0.2925d-2s+o(d)}$ assuming~\Cref{assumption: QRAM} and~\Cref{assumption: QRAM_filter}.
\end{theorem}

We describe the proof sketch here.
The formal proof of the above theorem is given in \Cref{app: lowerbound} along with a further finer formalization of the model, which turns out to be quite technical.

\begin{proof}[Proof sketch]
    Our proof strategy is to simulate each query of the quantum lattice sieving (or near-neighbor) algorithm classically, inspired by~\cite{HYY24}. Suppose that the quantum lattice sieving algorithm $A$, making $Q$ quantum queries, uses a QRAM of size at most $2^s$.
    For the quantum access to the input vectors in the relevant filter and indices, all the information possessed in these operations can be computed by $2^s$ corresponding classical queries in~\Cref{costmodel: cost}. For the inner product query, note that each register has at most $2^s$ different vectors, so every inner product $\inner{\vecx_i,\vecy_j}$ can be computed by $2^{2s}$ classical inner product queries. 

    Then, we construct the \emph{classical query} algorithm $B$ whose output is identical to one of $A$. This is done by collecting all the information mentioned above classically. 
    % However, $B$ \emph{simulates} the quantum behavior of $A$ only using the information from the classical queries. 
    This simulation may take time\footnote{Actual running time depends on the accuracy of the simulation.} much longer than $2^d$, but $B$ always makes classical queries to the oracle, thereby obeying the classical query complexity lower bound in~\Cref{thm: classical lower bounds}. 
    Note that each QRAM query can be simulated by $2^s$ classical queries, and a quantum inner product query can be simulated by $2^{2s}$ classical inner product queries.
    This means that the classical query complexity of $B$ is at most $2^{2s}\cdot Q$, which is lower bounded by $2^{0.2925d+o(d)}$. This establishes the lower bound in the statement.
\end{proof}
Plugging $s=0$ gives the following noteworthy corollary.
%, which is the main message of this section.
\begin{corollary}
    There is no quantum speed-up for the lattice sieving problem without QRAM in our model under~\Cref{assumption: QRAM} and \Cref{assumption: QRAM_filter}.
\end{corollary}

\section{Sieving Without QRAM}\label{sec:withoutQRACM}
% \jh{I wonder whether it is okay to use QRACM instead of QRAM specifically in this section.}
In this section, based on the algorithm in \cite{Heiser21}, we propose a quantum algorithm for lattice sieving without QRAM at the expense of using an exponential depth and number of qubits. 

\begin{lemma}\label{lem:cloest_vector_list}
    There is an algorithm that given $t$ classical lists $B_1,\ldots,B_t$
    of vectors, and in time $O(\sum_i|B_i|)$, creates a quantum circuit 
    $\mathcal{O}$ that satisfies,
    for any vector $\vec{w}$ and any index $1\leqslant i\leqslant t$:
    \[
        \mathcal{O}\ket{i}\ket{\vecw}\ket{0}\mapsto
            \ket{i}\ket{\vecw}\ket{\vecu}
    \]
    where $\vecu=\argmin_{\vecx\in B_i}\|\vecx-\vecw\|$
    is the closest vector to $\vecw$ in the list $B_i$.
    The circuit $\mathcal{O}$ has depth $O(\max_i|B_i|+\log t)$,
    size $O(\sum_i|B_i|)$ and width $O(t)$.
\end{lemma}
\begin{proof}
    We will first build a classical circuit that performs this operation,
    and then convert it to a quantum circuit in the standard way.

    Assume to simplify notations that each list $B_i$ has exactly
    $M$ vectors. Denote by $\vecu^{(i)}_1,\ldots,\vecu^{(i)}_M$
    the vectors in $B_i$. The circuit is described in 
    \Cref{fig:circuit_sieving_no_qram}. It uses two
    subcircuits:
    \begin{itemize}
        \item A standard multiplexer that takes $t$ inputs $a_1,\ldots,a_t$
            and an index $i$, and outputs $a_i$. This can be trivially implemented
            in depth $O(\log t)$, size $O(t)$ and width $O(t)$.
        \item A ``Compare $\vecu$'' circuit that takes two input vectors
            $\vecv$ and $\vecw$, and outputs $\argmin_{\vecx\in\{\vecu,\vecv\}}\|\vecx-\vecw\|$.
            In other words, it returns the vector that is closest to $\vecw$
            between the input $\vecv$ and the hardcoded vector $\vecu$.
            This can be implemented in constant
            depth, size, and width.
    \end{itemize}
    The idea of the circuit is very simple: for each $i$ in parallel,
    the circuit computes the closest vector to $\vec{w}$ in $B_i$, and
    selects at the end the right one based on requested index.
    For a given $i$, the circuit sequentially computes the closest
    vector to $\vec{w}$. Specifically, one can show by induction on 
    $1\leqslant m\leqslant M$ that
    for any $i$ and any $\vec{w}$, the output of
    the following circuit is
    $\argmin_{j=1,\ldots,m}\|\vec{u}_j^{(i)}-\vecw\|$:
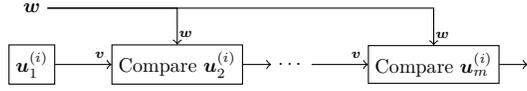
    \begin{center}
    \begin{tikzpicture}[scale=0.75, every node/.style={transform shape}]
        \node at (0,1) (dispatch) {$\vecw$};
        \node[draw,rectangle] at (0,0) (cmp_1)
            {$\vecu^{(i)}_1$};
        \node[draw,rectangle,right=1cm of cmp_1] (cmp_2) 
            {Compare $\vecu^{(i)}_2$};
        \node[right=0.5cm of cmp_2] (cmp_x) 
            {$\cdots$};
        \node[draw,rectangle,right=1cm of cmp_x] (cmp_M) 
            {Compare $\vecu^{(i)}_m$};
        \draw[->] (cmp_1) -> (cmp_2)
            node[pos=1,above left] {$\scriptstyle \vecv$};
        \draw[->] (dispatch) -| (cmp_2)
            node[pos=1,above,xshift=1ex] {$\scriptstyle \vecw$};
        \draw[->] (cmp_2) -- (cmp_x);
        \draw[->] (cmp_x) -- (cmp_M)
            node[pos=1,above left] {$\scriptstyle \vecv$};
        \draw[->] (cmp_M.east) -- +(0.5,0);
        \draw[->] (dispatch) -| (cmp_M)
            node[pos=1,above,xshift=1ex] {$\scriptstyle \vecw$};
    \end{tikzpicture}
    \end{center}
    Overall, the circuit has depth $O(M)+\log t$ since the longest
    path goes through $M$ comparisons and the multiplexer. The width
    of the circuit is $O(t)$ since the $t$ comparisons chains all
    happen in parallel. The size of the circuit is $O(Mt)$.
    It is clear that there is an algorithm that can construct this circuit
    is time at most $O(Mt)$ (the overall size).
    If the lists $B_i$ have different sizes, the depth depends on
    the biggest $B_i$ and the total size is the sum of the sizes of $B_i$.

    We can convert this circuit to a quantum circuit using standard techniques. The resulting
    circuit has depth $O(M+\log t)$, size $O(Mt)$ and width $O(t)$.
    Furthermore, the algorithm that does the conversion runs in time
    linear in the size of the circuit.

    Denote by $\mathcal{O}$ the quantum circuit obtained above
    where we add a register for the output. By
    the description of the circuit, it is clear that for any vector
    $\vecw$ and any index $i$,
    \[
        \mathcal{O}\ket{i}\ket{\vecw}\ket{0}\mapsto
            \ket{i}\ket{\vecw}\ket{\vecu}
    \]
    where $\vecu=\argmin_{\vecx\in B_i}\|\vecx-\vecw\|$ is the closest vector to $\vecw$
    in the $B_i$.
\end{proof}
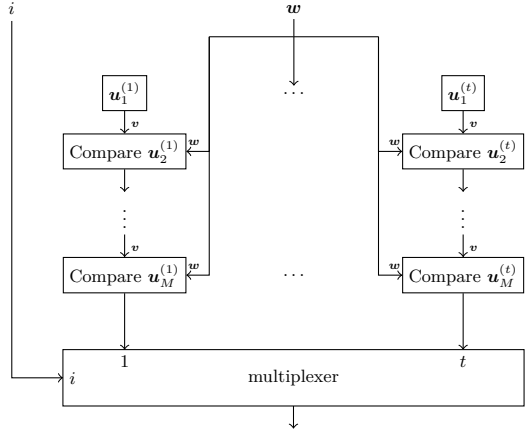
\begin{figure}[ht]
    \centering
    \begin{tikzpicture}[scale=0.75, every node/.style={transform shape}]
        \node at (3,1.5) (input_w) {$\vecw$};
        \node at (-2,1.5) (input_i) {$i$};
        \coordinate (dispatch_w) at (3,1);
        \coordinate (dispatch1) at (1.5,1);
        \coordinate (dispatch2) at (4.5,1);
        \draw (dispatch1) -- (dispatch_w);
        \draw (dispatch2) -- (dispatch_w);
        \draw (input_w) -- (dispatch_w);
        \foreach \x/\idx/\dispatch/\shift in {
            0/1/dispatch1/1ex,
            6/t/dispatch2/-1ex
        } {
            \node[draw,rectangle] at (\x,0) (cmp_\idx_1)
                {$\vecu^{(\idx)}_1$};
            \node[draw,rectangle,below=0.4cm of cmp_\idx_1] (cmp_\idx_2) 
                {Compare $\vecu^{(\idx)}_2$};
            \node[below=0.4cm of cmp_\idx_2] (cmp_\idx_x) 
                {$\vdots$};
            \node[draw,rectangle,below=0.4cm of cmp_\idx_x] (cmp_\idx_M) 
                {Compare $\vecu^{(\idx)}_M$};
            \draw[->] (cmp_\idx_1) -> (cmp_\idx_2)
                node[pos=1,above right] {$\scriptstyle \vecv$};
            \draw[->] (\dispatch) |- (cmp_\idx_2)
                node[pos=1,above,xshift=\shift] {$\scriptstyle \vecw$};
            \draw[->] (cmp_\idx_2) -> (cmp_\idx_x);
            \draw[->] (cmp_\idx_x) -> (cmp_\idx_M)
                node[pos=1,above right] {$\scriptstyle \vecv$};
            \draw[->] (\dispatch) |- (cmp_\idx_M)
                node[pos=1,above,xshift=\shift] {$\scriptstyle \vecw$};
        }
        \node at (3,0) (top_dots) {$\ldots$};
        \node at (cmp_1_M -| top_dots) (bot_dots) {$\ldots$};
        \draw[->] (dispatch_w) -- (top_dots);

        \coordinate (mux_top_left) at ($(cmp_1_M.south west)+(0,-1)$);
        \coordinate (mux_bot_right) at ($(cmp_t_M.south east)+(0,-2)$);
        \draw (mux_top_left) rectangle (mux_bot_right)
            node[pos=0.5] {multiplexer};
        \draw[->] (cmp_1_M) -- (cmp_1_M |- mux_top_left)
            node[below] {$1$};
        \draw[->] (cmp_t_M) -- (cmp_t_M |- mux_top_left)
            node[below] {$t$};
        \coordinate (mux_bot_left) at (mux_top_left |- mux_bot_right);
        \draw[->] ($(mux_bot_left)!0.5!(mux_bot_right)$) -- ++(0,-0.4);
        \draw[->] (input_i) |- ($(mux_top_left)!0.5!(mux_bot_left)$)
            node[right] {$i$};
    \end{tikzpicture}
    \caption{\label{fig:circuit_sieving_no_qram}Circuit diagram that outputs the vector in $B_i$
    closest to the input vector $\vecw$. Every vector in $B_i$ is hardcoded in the circuit,
    and serially compared to $\vecw$. See \Cref{lem:cloest_vector_list}.}
\end{figure}

\begin{theorem}\label{th:quantum_lsf_no_qram}
    There is a quantum lattice sieving algorithm that,
    given a list of $n$ random input vectors and $\alpha,\beta \in (0,1)$, 
    outputs $\Omega(n)$ reduced vectors based on the LSF with $t=\mathcal{W}_{d}(\alpha, \beta, \pi/3)^{-1}$ filters.
    The running time of the algorithm is
    \begin{equation}
        nt\cdot \mathcal{C}_{d}(\beta)
            + n\cdot \sqrt{ t\cdot \mathcal{C}_{d}(\alpha)}\cdot
            \max(1,n\cdot \mathcal{C}_{d}(\beta)).
        \label{eq:quantum_lsf_no_qram_unopt}
    \end{equation}
    The algorithm uses no QRAM but evaluates a quantum circuit of
    depth $O(M+\log t)$, size $O(Mt)$ and width $O(t)$.
\end{theorem}
\begin{proof}
    As in \Cref{section: quantum_lsf_sieving}, 
    the algorithm first prepare $t$ LSF filters
    $\mathcal F_1,\ldots,\mathcal F_t$. For each $i$, we compute
    the classical list $B_i$ of vectors inside $\mathcal F_i$.

    We now apply \Cref{lem:cloest_vector_list} to build a quantum circuit
    $\mathcal{O}$ such that for any vector $\vec{w}$ and any index 
    $1\leqslant i\leqslant t$:
    \[
        \mathcal{O}\ket{i}\ket{\vecw}\ket{0}\mapsto
            \ket{i}\ket{\vecw}\ket{\vecu}.
    \]
    where $\vecu=\argmin_{\vecx\in B_i}\|\vecx-\vecw\|$
    is the closest vector to $\vecw$ in the list $B_i$.
    The circuit $\mathcal{O}$ has depth $O(\max_i|B_i|+\log t)$,
    size $O(\sum_i|B_i|)$ and width $O(t)$.

    By \Cref{lem:heiser_sampler}, there is a (classical) sampler
    that can return, for any given $\vec{w}$, a random
    (pseudo) $\alpha$-close filter to $\vec{w}$ in $\poly(d)$ time.
    This sampler only requires $2^{o(d)}$ preprocessing time.
    We can turn this sampler into a deterministic algorithm that
    takes as input $R$ ``random coins''. Denote by
    $\mathcal{A}(\omega,\vecw)$ such a run where $\omega$ denotes the random
    coins. By \Cref{lem:heiser_sampler}, we have
    $R=O(\log(d)\cdot\log(t\cdot C_d(\alpha)))$.

    Consider the quantum oracle $\mathcal{O}'$
    that on input $\omega\in\{0,1\}^R$ and $\vecw$:
    \begin{itemize}
        \item compute $i\gets \mathcal{A}(\omega,\vecw)$,
        \item returns $\mathcal{O}(i,\vecw)$.
    \end{itemize}
    It is clear by the properties of $\mathcal{A}$
    that $\mathcal{O}'(\omega,\vecw)$ returns a vector $\vecu\in B_j$
    for some $j$ such that $\vecw\in B_j$.
    Overall, the sieving algorithm looks like this:
    \begin{itemize}
        \item Go through all $n$ vectors and put each of them in
            the filters $\mathcal F_i$ to which they belong.
        \item Apply \Cref{lem:cloest_vector_list} to build $\mathcal{O}$.
        \item Use \cite{Heiser21} to build $\mathcal{O}'$.
        \item For each vector $\vecw$ in the input list:
            run a quantum minimum finding algorithm using $\mathcal{O}'$
                    to find $\omega\in\{0,1\}^R$ that minimizes
                    $\norm{\mathcal{O}'(\omega,\vecw)-\vecw}$
                and get the closest vector to $\vec{w}$ in the list.
                We then check if this vector forms a reduced pair and add it to the output
                list if that's the case.
            Based on the analysis done in the previous sections, we expect a constant number
            of vectors in the list to form a reduced pair with $\vec{w}$. Therefore by finding the closest vector
            to $\vec{w}$, we are sure to find a reduced pair using $\vec{w}$ if one exists.
    \end{itemize}
    In the first step, each vector
    belongs to $t\cdot \mathcal{C}_{d}(\beta)$ filters on average so
    the complexity is $nt\cdot \mathcal{C}_{d}(\beta)$. Furthermore,
    the lists $B_i$ have an average size of $M=n\cdot \mathcal{C}_{d}(\beta)$.

    In the second step, the complexity of building $\mathcal{O}$ is
    $O(\sum_i|B_i|)=O(Mt)=O(nt\cdot \mathcal{C}_{d}(\beta))$.
    The circuit $\mathcal{O}$ has depth $O(\max_i|B_i|+\log t)=O(M+\log t)$,
    size $O(Mt)$ and width $O(t)$.

    In the third step, the complexity of building $\mathcal{O}'$ is
    dominated by the preprocessing cost of \cite{Heiser21}'s sampler which
    is $2^{o(d)}$.
    
    In the final step, the complexity is $n$ times $2^{R/2}$
    multiplied by the cost of the oracle $\mathcal{O}'$.
    The cost of $\mathcal{A}$ is $\poly(d)$ so this is essentially
    the cost of $\mathcal{O}$.
    Each evalutation of $\mathcal{O}$ costs the depth of its circuit
    which is $O(\max_i|B_i|+\log t)=O(M+\log t)$. The sampler $\mathcal{A}$
    returns a uniform sample in the set of ``good filters''
    which is of average size $S=t\cdot \mathcal{C}_{d}(\alpha)$.
    Furthermore, $\mathcal{A}$ only requires
    $R=O(\log(d)\log(S))$ random coins. Since $\mathcal{A}$ is a uniform
    sampler in a set of size $S$, for a given $\vecw$,
    $\mathcal{A}(\cdot,\vecw)$
    takes each of the $S$ possible outputs $2^R/S$ times.
    It therefore follows
    that $\mathcal{O}'(\cdot,\vecw)$ takes each of the $S$ possible output $k=2^R/S$ times.
    As a result, running the minimum finding algorithm on $\mathcal{O}'(\cdot,\vecw)$
    takes time $O(\sqrt{2^R/k})=O(\sqrt{S})$.
    Hence, the time complexity is
    \begin{equation}
        nt\cdot \mathcal{C}_{d}(\beta)
        + n\cdot \sqrt{ t\cdot \mathcal{C}_{d}(\alpha)}\cdot
        \max(1,n\cdot \mathcal{C}_{d}(\beta)).
    \end{equation}
    Here, the $\max$ is necessary to handle the case where there might be
    less than one vector in each bucket on average, because we still have
    to pay $O(1)$ just to examine the bucket.
    The number of qubits used is $O(t+R)$.
\end{proof}

We optimize the time complexity of \Cref{th:quantum_lsf_no_qram}
using the following formula: 
\[ 
 t=\left(1-\frac{4}{3}(\alpha^2-\alpha\beta+\beta^2) \right)^{-d/2}, 
\quad \mathcal{C}_{d}(\alpha)=(1-\alpha^2)^{d/2},
\quad n=(4/3)^{d/2}.
\]
For each value of $t$, we compute the optimal time complexity by finding the optimal $\alpha$ and $\beta$ such that $t$ satisfies
the equation above.
Experimentally, we observe that the optimal complexity is close to
\[
    2^{0.414d-0.655\log_2(t)}=\frac{2^{0.414d}}{t^{0.655}}
\]
We observe a regime (see \Cref{fig:sieve_no_qram_graph}) where we can achieve a better time complexity than the best classical sieving algorithm while using an exponential number of qubits. Note than we only draw the graph up to $t=2^{0.207d}$ qubits, since there is a quantum algorithm in \cite{KMPM19} that solves SVP in time $2^{0.1037d+o(d)}$ with $2^{0.207d+o(d)}$ qubits. \cite{KMPM19} does not seem to work with less
than $2^{0.207d+o(d)}$ qubits because it requires as many qubits as there are elements in the list,
therefore it is incomparable with the above algorithm.
Furthermore, the computational model of \cite{KMPM19}
is slightly different since it is more akin to a
massively parallel/distributed system.

\begin{figure}
    \centering
    \begin{tikzpicture}[scale=0.75, every node/.style={transform shape}]
        \begin{axis}[
            xlabel={$\tfrac{1}{d}\log_2(t)=\tfrac{1}{d}\log_2(\text{qubits})$},
            xtick distance=0.05,
            minor x tick num=10,
            ylabel={$\tfrac{1}{d}\log_2(\text{complexity})$},
            ytick distance=0.02,    % jh) originally 0.01
            height=5cm,
            width=9cm       % jh) height & width added
        ]
            \addplot[smooth,blue] table {sieve.csv};
            \addplot[domain=0:0.207519,red] {0.292};
        \end{axis}
    \end{tikzpicture}
    \caption{\label{fig:sieve_no_qram_graph}Complexity
        of the sieve algorithm with no QRACM as a function of
        the number of qubits. The red line corresponds to the best
        classical complexity.
        }
\end{figure}
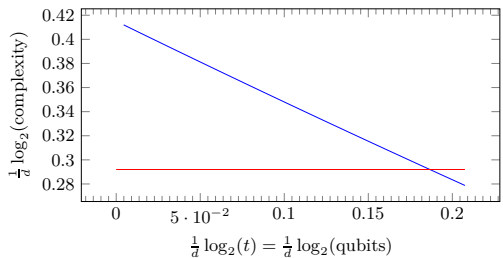

\section*{Acknowledgement}
The work of Beomgeun Cho, Taehyun Kim, and Jeonghoon Lee was supported by Samsung Electronics Co., Ltd(IO221213-04119-01).
The work of Yixin Shen was funded by the EPSRC grant EP/W02778X/2 and the France 2030 program managed by the French National Research Agency under grant agreement ANR-22-PETQ-0008 PQ-TLS.

\bibliographystyle{splncs04}
\bibliography{ref}
\newpage
\appendix
\section{Omitted Proofs}\label{app: omitted proofs in QRAM bdd}
\begin{proof}[Proof of \Cref{lemma: limited_qracm_grover}]
    Write $S_i:=[i\cdot S + 1, (i+1) \cdot S]$ for brevity.
    Write $X=\{x_1,...,x_{M}\}$ (order can be arbitrary) and its subsets $X_i:= \{x_{i \cdot S+1},...,x_{(i+1) \cdot S}\} = \{x_j: j \in S_i\}\subset X$ of size $S$ for $i=0,...,\frac MS-1$.

    We define the algorithm $A$ as follows: For each $0\le i\le \frac MS-1$, $A$ stores $X_i$ in QRAM. Define 
    \[Init_i:\ket{0}\mapsto 
    \frac{\left( 
        % \sum_{j=i\cdot S +1}^{(i+1) \cdot S} \ket{j,x_{j}}
        \sum_{j \in S_i} \ket{j,x_{j}}
    \right) }{\sqrt S} \] 
    using QRAM similarly to~\cref{eqn: initwithfullQRAM}. It applies quantum amplitude amplification to $O_P=I \otimes O_f$\footnote{Here, $I$ acts on the index register.} and $Init_i$ for $O(\sqrt S)$ times, and measures the result to obtain $( j^*, x_{j^*})$, and check if $f(x_{j^*})=1$. If true, it outputs $x_{ j^*}$ and halts. If there is no such $x$ for all $i,$ it returns $\bot$.

    The running time of $A$ requires at most $O\left( \frac MS \cdot \sqrt S \right)=O\left( \frac{T}{\sqrt S}\right)$ evaluations of $f$ as we want, because we limit the number of iterations to $\sqrt S$ for each quantum amplitude amplification. 
    
    Next, we argue the correctness of the algorithm. Suppose that there exists $x^* \in X_i$ such that $f(x^*)=1$. Then, \Cref{thm:QAA} implies that the $i$-th quantum amplitude amplification can find $x \in X_i$ such that $f(x)=1$ with a sufficiently high probability.    
    % The search space $X$ can be divided into $2^{n-s}$ blocks, each having $2^s$ entries. After loading a block of data to the QRACM, we can apply Grover search to find a solution (or no solution signal) with $O(2^{s/2})$ queries. By repeating the same procedure for $2^{n-s}$ blocks, we get the complexity of $O(2^{n-s/2})$, and the cost of intermediate QRAM overwriting can be ignored with the constant time assumption above.
\end{proof}

\begin{proof}[Proof of \Cref{lemma: limited_qracm_grover_pair}]
    Write $S_i:=[i\cdot S + 1, (i+1) \cdot S]$ for brevity.
    Let $X=\{x_1, \cdots, x_{M_1}\}, Y=\{y_1, \cdots, y_{M_2}\}$, and define subsets 
    $
    % X_i := \{x_{i\cdot S+1}, \cdots, x_{(i+1)\cdot S}\}\subseteq X,\ 
    X_i := \{x_{k}: k \in S_i\}\subseteq X,\ 
    Y_j :=\{ y_{\ell}: \ell \in S_j\}\subset Y$ of size $S$ for $i=0, \cdots, \frac {M_1}S-1,\ j=0, \cdots, \frac{M_2}S-1$.
    
    With two QRAMs of size $S$, the algorithm $A'$ runs amplitude amplification, similar to the proof in \Cref{lemma: limited_qracm_grover}:
    After storing $X_i$ and $Y_j$ in each QRAM, define two operations; 
    \begin{align}\label{eqn: init_ij}
        Init_{i, j} := |0\rangle \mapsto \frac
        {\sum_{k \in S_i,\ell \in S_j}\ket{k,\ell, x_{k}, y_{\ell}}}{S}
    \end{align}
    with access to the QRAMs, and $O_P:=I\otimes O_{f}$. As each $X_i, Y_j$ are randomly selected, the expected number of solutions in $X_i \times Y_j$ is $E:=\frac{K\cdot S^2}{M_1 \cdot M_2} $. 
    The algorithm behaves differently, conditioned on the expected number of solutions.
    % There are two cases for the expected number of solutions.\minki{solutions? iterations?}
    % \minki{I found the following descriptions are unclear. What is the number of iterations here?}
    \begin{enumerate}
        \item If the expected number of solutions in the current QRAM is at least 1 (i.e., $KS^2 \ge M_1 M_2$), then we run the amplitude amplification multiple times to find most solutions. The number of calling $Init_{i,j}$ in~\Cref{eqn: init_ij} and $f$ is $E \cdot \sqrt{\frac{S^2}{E}}  = S \cdot \sqrt{ E} = \frac{\sqrt K\cdot S^2}{\sqrt{M_1 \cdot M_2}}$.
        \item Otherwise, it needs $\sqrt{S^2}$ iterations to check if a solution exists in the current QRAM.
    \end{enumerate}

    There are $\frac{M_1 \cdot M_2}{S^2}$ pairs of $(X_i, Y_j)$, therefore the total number of required iterations is $\sqrt{M_1 M_2 K}$ for the first case, and $\frac{M_1 M_2}S$ for the second case.
    For the correctness of $A'$, \Cref{thm:QAA} implies that the measurement output gives the solution pair $(x^*, y^*)\in X\times Y$ such that $f(x^*, y^*)=1$ with a sufficiently high probability.
\end{proof}

\subsection{The bounded QRAM search lower bound}\label{subsec: QRAM bound search proof}
In this section, we give the proof of~\Cref{thm: lower_bound bounded QRAM search}. The proof uses the recording random functions~\cite{Zha19,CGKSW23}. The proof of the main lemma (\Cref{lem: progress}) is inspired by the proofs in~\cite{HLS24}.

\subsubsection{Preparation: Bernoulli random functions.}

The proof requires the recording technique for the Bernoulli random function introduced in~\cite{CGKSW23}. We give a brief introduction below.

Let $|X|=2^m$.
We call $f:\{0,1\}^m \to \{0,1\}$ be a Bernoulli random function with parameter $0<p<1$ if $f(x)=1$ holds with probability $p$ independently.
We denote the distribution of the Bernoulli random function by $B_{m,p}$, and $\alpha_f$ denotes the probability that the function $f$ is sampled from $B_{m,p}$.
It is well known that an algorithm having oracle access to a Bernoulli random function is identical to having oracle access to the purified Bernoulli random function that is defined by
\[
\sum_{f\in B_{m,p}} \sqrt{\alpha_f}_F \ket{f}_F
\]
where $F$ is a $|X|$-qubit register. The standard query is computed by
\begin{align}\label{eqn: Bquery}
    \stdBO:\ket{x,y}\otimes\sqrt{\alpha_f}_F \ket{f}_F \mapsto(-1)^{y\cdot f(x)}\ket{x,y}\otimes\sqrt{\alpha_f}_F \ket{f}_F
\end{align}

We define the generalized Hadamard operation on a single qubit register
\[
U_p:\ket{b} \mapsto \sqrt{1-p} \ket{b} + (-1)^b \sqrt{p} \ket{b\oplus 1}
 \text{ or }U_p = \begin{bmatrix}
        \sqrt{1-p} & -\sqrt p\\
        \sqrt p & \sqrt{1-p}
    \end{bmatrix}.
\]
Note that $\cU_p:= \otimes_{x\in\{0,1\}^m} U_p^x$ where $U_p^x$ denotes $U_p$ applied on the $x$-th qubit gives
\[
\cU_p\ket{0^{|X|}}_F\mapsto \sum_{f\in B_{m,p}} \sqrt{\alpha_f}_F \ket{f}_F.
\]
We define the dual query by
\begin{align}
    \compBO:= (I\otimes \cU_p^\dagger) \cdot \stdBO \cdot (I\otimes \cU_p) 
\end{align}
with the initial state $\ket{0}^{\otimes |X|}.$ It is not hard to see that the output of any algorithm having access to $\stdBO$ (with the initial state $\sum_{f\in B_{m,p}} \sqrt{\alpha_f}_F \ket{f}_F$) and having access to $\compBO$ (with the initial state ($\ket{0}^{\otimes |X|}$) is identical.

We can compute the progress of the overall states under $\compBO$ using the following lemma, which is a Bernoulli analog of~\cite[Lemma 4.3]{CHFL21}. The proof follows from the straightforward calculation.
\begin{lemma}\label{lem: RBO query}
    The map $\compBO$ operates as follows, where $\ket{\cdot}_x$ denotes the $x$-th qutrit of $F$.
    \begin{align*}
        &\ket{x,+}\otimes \ket{b}_x &&\mapsto \ket{x,+}\otimes \ket{b}_x\text{ for arbitrary $b$}\\
        &\ket{x,-}\otimes \ket{0}_x&&\mapsto
        \ket{x,-}\otimes((1-2p)\ket{0}_x - 2\sqrt{p(1-p)}\ket{1}_x)\\
        % \ket{x,-}\otimes \ket{\hat 0}\mapsto \sqrt{1-p} \ket{x,-}\otimes \ket{0} - \sqrt p \ket{x,-}\otimes \ket{1} 
        % \\
        % &=\ket{x,-}\otimes (\sqrt{1-p} \ket{0} - \sqrt{p}\ket{1}) =\ket{x,-} \otimes \left((1-2p)\ket{\hat 0} -2\sqrt{p(1-p)} \ket{\hat 1}\right)
        % \\&\mapsto\ket{x,-}((1-2p)\ket{0} - 2\sqrt{p(1-p)}\ket{1})\\
        &\ket{x,-}\otimes\ket{1}_x&&\mapsto
        % \ket{x,-}(\sqrt{1-p}\ket{1}-\sqrt{p}\ket0)\\
        % &\mapsto-\ket{x,-}(\sqrt{1-p}\ket{1}+\sqrt{p}\ket0)=
        \ket{x,-}\otimes(-2\sqrt{p(1-p)}\ket{0}_x + (2p-1) \ket{1}_x)
    \end{align*}
\end{lemma}

\subsubsection{Proof of the lower bound.}

Let $A$ be an algorithm described in~\Cref{thm: lower_bound bounded QRAM search}. Recall the oracle can be written as a Bernoulli random function $f:\{0,1\}^{|X|}\to \{0,1\}$ with parameter $p$, and $S$ be the bound of the elements in the computational basis, and $q$ be the number of queries. We consider the overall states of the algorithm and the oracle $\compBO$ initialized by
\[
\ket{\phi^{(0)}}_{AF}\ket{0}_A \otimes \ket{0^{|X|}}_F.
\]
Similarly, we write $\ket{\phi^{(t)}}_{AF}$ to denote the overall state after the $t$-th query.
In the final step, we apply a projection $\Pi:=\sum_x \ketbra{x}_o \otimes \ketbra{1}_x$ where $o$ denotes the output register of $A$, and $\ketbra{1}_x$ checks if the output is correct; the unspecified registers are not changed. 
For the final state $\ket{\phi}_{AF}$, the success probability is written by
\begin{equation}
    p_{success}:=\left\|\Pi (I\otimes \cU) \ket{\phi^{(q)}}_{AF}\right\|^2.
    % =\left\|\Pi_x (I\otimes U_p) \phi_{ox}\right\|^2
    % = \sum_x \left\| 
    %     \ketbra{x,1}_{ox} \cdot (I\otimes U_p^x) \phi_{ox}
    % \right\|^2 
\end{equation}
By~\cite[Lemma 3.11]{CGKSW23}, we have
\begin{equation}\label{eqn: connecting}
    \left\|\Pi (I\otimes \cU) \ket{\phi^{(q)}}_{AF}\right\| \le \sqrt p + \left\|\Pi \ket{\phi^{(q)}}_{AF}\right\|.
    % \left |\left\|\Pi (I\otimes \cU) \ket{\phi^{(q)}}_{AF}\right\|
    % - \left\| 
    %     P \ket{\phi^{(q)}}_{AF}
    % \right\| \right | = O(\sqrt{p\cdot \ell}).
\end{equation}
% \minki{I think this ``a standard argument'' is too sketch...}
% A standard argument (see e.g.,~\cite[Lemma 5]{Zha19} or~\cite[Lemma 4.2]{CHFL21}) shows that if a projector $P$ is defined by the condition on the $\ell$ entries of the database, then
% \begin{equation}\label{eqn: connecting}
%     \left |\left\|P (I\otimes \cU) \ket{\phi^{(q)}}_{AF}\right\|
%     - \left\| 
%         P \ket{\phi^{(q)}}_{AF}
%     \right\| \right | = O(\sqrt{p\cdot \ell}).
% \end{equation}
We will give the upper bound of the progress measure $p_t:=\left\| 
    \Pi\ket{\phi^{(t)}}_{AF}
\right\|^2$ and relate it to the success probability.
Using the following lemma and~\Cref{eqn: connecting}, we have
$\sqrt{p_{success}} \le \sqrt{p_q} +\sqrt p$ which implies
\[
p_{success}=O(\sqrt S \cdot pq)
\]
as desired.
\begin{lemma}\label{lem: progress}
    The following hold:
    \begin{enumerate}
        \item $p_0=0.$
        \item $p_{t+1} \le p_{t} + c\cdot \sqrt S \cdot p$ for some constant $c>0$.
        % \item $\sqrt{p_{success}} \le \sqrt{p_q} + \sqrt p$
    \end{enumerate}
\end{lemma}
\begin{proof}
The first item is obvious: Since there is no $1$ in the $F$ register in the initial state, $p_0=0$.

To prove the second item, we define the projectors $\Pi^c_Z:= \otimes_{x \in Z} \ketbra{0}_x$, and
\[
\Pi_{Z} := \sum_{x \in Z} \ketbra{x}_o \otimes \ketbra{1}_x \otimes \Pi^c_{Z^c}, \Pi_{\lnot Z} := \Pi-\Pi_Z, \text{ and }\Pi^c:=I-\Pi
\]
for a subset $Z \subset \{0,1\}^{|X|}$. Intuitively, $\Pi_Z$ projects to the state where the entry $1$ is found only in some $x\in Z$, and all other registers contain 0.

Applying a unitary on the register $A$ does not change the progress measure. Thus, it suffices to prove the inequality for $\ket{\phi^{(t+1)}}_{AF} = \compBO\ket{\phi^{(t)}}_{AF}$ for any 
\[
\ket{\phi^{(t)}}_{AF} = \sum_{x,y \in X_{t} \times \{+,1\},D} \alpha_{xyD} \ket{x,y} \otimes \ket{D}_F \otimes \ket{\phi_{xyD}}
\]
where $X_t \subset \{0,1\}^{|X|}$, which exists due to the assumption.
Then, we have
\begin{align}
    p_{t+1} = \left\| 
        \Pi \ket{\phi^{(t+1)}}
    \right\|^2 = \left\| 
        \Pi \cdot \compBO \cdot (\Pi_{\lnot {X_t}} + \Pi_{X_t} + \Pi^c)\ket{\phi^{(t)}}
    \right\|^2 
\end{align}
where we use $\Pi_{\lnot {X_t}} + \Pi_{X_t} + \Pi^c = \Pi + \Pi^c  =I$. This equals to
\begin{align}
    \left\| 
        \Pi \cdot \compBO\cdot \Pi_{\lnot {X_t}} \ket{\phi^{(t)}} \right\|^2
        +\left\|\Pi \cdot \compBO
        \cdot (\Pi_{X_t} + \Pi^c)\ket{\phi^{(t)}}
    \right\|^2 
\end{align}
because the database of $\Pi \cdot \compBO\cdot \Pi_{\lnot {X_t}} \ket{\phi^{(t)}} $ always contains $\ket{1}_x$ for some $x \notin X_t$, but the database of $\Pi \cdot \compBO\cdot (\Pi_{X_t} + \Pi^c)\ket{\phi^{(t)}}$ contains $\ket{1}_x$ only for $x\in X_t$. This is bounded above by
\begin{align}
    &\le \left\| 
        \Pi_{\lnot {X_t}} \ket{\phi^{(t)}} 
    \right\|^2
    +\left(\left\|\Pi_{X_t}\ket{\phi^{(t)}}\right\| +\left\| 
    \Pi \cdot \compBO
        \cdot \Pi^c\ket{\phi^{(t)}}
    \right\|\right)^2 \\
    &= \left\| 
        \Pi \ket{\phi^{(t)}} 
    \right\|^2
    +\left\| 
    \Pi \cdot \compBO
        \cdot \Pi^c\ket{\phi^{(t)}}
    \right\|^2\\&
    +2\left\|\Pi_{X_t}\ket{\phi^{(t)}}\right\|\cdot\left\| 
    \Pi \cdot \compBO
        \cdot \Pi^c\ket{\phi^{(t)}}
    \right\|.
\end{align}
Using~\Cref{lem: RBO query}, it is easy to derive that $\|\Pi \cdot \compBO \cdot \Pi^c \ket{\phi} \|^2 \le p$ for any $\ket{\phi}$. 

By modifying the proof of~\cite[Lemma 3.11]{CGKSW23} (by just changing the role of the primal and dual domain and applying $\Pi_{X_t}$), we have
\[
\left\|\Pi_{X_t}\ket{\phi^{(t)}}\right\| \le \left\|\Pi_{X_t} (I\otimes \cU)\ket{\phi^{(t)}}\right\| + \sqrt p.
\]
We have $\left\|\Pi_{X_t} (I\otimes \cU)\ket{\phi^{(t)}}\right\|^2 $ is bounded by the probability that the algorithm finds a solution in $X_t$. This is again bounded by 
$\le S\cdot p$ due to $|X_t| \le S$, which is the probability that there exists a solution in $X_t$ for the random Bernoulli function.
This implies that
$\left\|\Pi_{X_t}\ket{\phi^{(t)}}\right\| =O( \sqrt{S \cdot p}).$ 

Plugging this, we have the final upper bound
\[
p_t + p + 2c \sqrt S \cdot p 
\]
for some constant $c>0$. ($c=10$ suffices.)
This concludes the proof.
\end{proof}

\section{Proof of Bounded QRAM Lower Bounds}\label{app: lowerbound}
This section proves~\Cref{thm:qsieve_lowerbound}. 
To this end, we formally present the model of algorithms and costs, and the lower bounds of quantum sieving with a bounded QRAM.

\subsection{Black-box near-neighbor algorithm and classical lower bound}
In our model, a \emph{black-box near-neighbor} algorithm $A$ interacts with an oracle that keeps a list $L=(\vecx_1,...,\vecx_n)$ of size $n=2^{0.2075d+o(d)}$ and the sets $A_i,B_i$ for $ I \in [t]$ initialized by empty sets. The oracle decides random vectors $\vecz_1,...\vecz_t$ from $\Sd$ uniformly at random and define
$Q_i$ and $U_i$ by spherical caps as in~\Cref{eqn: optchoice}.

The algorithm $A$ works over a working register $\rW$, a query register $\rQ$, and an arbitrarily long table register $\rT$ and $\rX$. The registers $\rW$ and $\rQ$ are initialized to $\ket{0...0}$ and $A$ is allowed to apply an arbitrary unitary on $\mathbf {WQ}$ during its execution. The registers $\rX$ and $\rT$ are initialized by $\ket{\vecx_1,...,\vecx_n}$ and $\ket{0,0,...}$. The algorithm $A$ is allowed to make the following types of queries:
\begin{itemize}
    \item[$\bullet$] \emph{Insertion.} Measure the query register and interpret it as an element $(C,i,j)\in\{A,B\}\times [t] \times [n]$. The oracle updates $C_i \gets C_i \cup \{\vecx_j\}$.
    \item[$\bullet$] \emph{Sampling vectors from filters.} Apply the following operation on $\rQ$ and $\rT$ that works as follows in a computational basis:
    \begin{align}
        \ket{C,i,j,k}_\rQ \ket{...,0,...}_\rT \mapsto \ket{C,i,j,k}_\rQ \ket{..., \vecc_j,...}_\rT
    \end{align}
    where $C\in \{A,B\}$, $0$ is the $k$-th entry of $\rT$ and $\vecc_j$ is the $j$-th element of $C_i$; if no such $\vecc_j$, define $\vecc_j=\bf 0$.
    \item[$\bullet$] \emph{Copying from $\rX$ to $\rT$} Apply the following operation on $\rQ,\rX$ and $\rT$:
    \begin{align}
        \ket{i,j}_\rQ \ket{...,\vecx_i,...}_\rX\ket{...,0,...}_\rT \mapsto \ket{i,j}_\rQ \ket{...,\vecx_i,...}_\rX\ket{...,\vecx_i,...}_\rT
    \end{align}
    where $0$ is the $j$-th entry of $\rT$.
    \item[$\bullet$] \emph{Sampling filters from vectors.} Apply the following operation on $\rQ$ and $\rX$ that works as follows in a computational basis:
    \begin{align}
        \ket{Z,i,0,k}_\rQ \ket{...,\vecx_k,...}_\rX \mapsto \ket{Z,i,\ell_i,k}_\rQ \ket{...,\vecx_k,...}_\rX
    \end{align}
    where $Z\in\{C,Q\}$, $\vecx_k$ is the $k$-th entry of $\rX$, $\{Z_{\ell_1},...,Z_{\ell_u}\}$ be the set of relevant ($\alpha$- or $\beta$- depending on $Z$) filters such that $\ell_1<...<\ell_u$, and set $\ell_i =0 $ if $i>u$. If the third register of $\rQ$ is non-zero, it does nothing.
    \item[$\bullet$] \emph{Inner product.} Measure the second and third registers of $\rQ$.\footnote{This measurement is to fix the indices of the inner product, which corresponds to the operation given in~\Cref{costmodel: quantum_cost}. If we do not apply the measurement, it requires quantum RAM for quantum data.} Apply the following operation on $\rQ$ and $\rT$ that works as follows in a computational basis:
    \begin{align}
        \ket{r,k,\ell}_\rQ \ket{...,\vect_k,...,\vect_\ell,...}_\rT \mapsto
        \ket{r\oplus b,k,\ell}_\rQ \ket{...,\vect_k,...,\vect_\ell,...}_\rT
    \end{align}
    where $\vect_k,\vect_\ell$ denote the $k,\ell$-th entries of $\rT$, and $b=1$ if $\inner{\vect_k,\vect_\ell}\ge 1/2$, otherwise $b=0$.
\end{itemize}
Except for the insertion, the inverse operations also can be queried. 
Note that each entry of the registers $\rT,\rX$ always contains $0$ or a vector in $L$. Also note that the register $\rX$ is always unchanged from the initial classical state.

If we consider the \emph{classical} near-neighbor algorithms, $\rQ$ is measured before applying the unitary. (We allow the other parts to be coherent.) 
The black-box algorithm cannot apply any other operation on $\rT,\rX$ besides the oracle queries, and comparison between two registers. This implies that the algorithm does not know anything about the input vectors except the queries: the indices of $\alpha$- or $\beta$-relevant filters, the closeness of pairs,

We count the number of the last two types of queries (sampling filters and inner products) as the complexity measure. We also note that the black-box near-neighbor algorithm follows the algorithm outlined in~\Cref{alg:hashNNS}, especially \textsc{Preprocess}.
We formally assume the following goal of the algorithm.
\begin{definition}\label{def: method_purpose}
    % It is promised that $t \cW_d(\alpha,\beta,\pi/3)= \Omega(1)$\footnote{This is }
    After \textsl{Preprocess}, the hash-based near-neighbor algorithm is
    asked to find each tuple $(\vecx,\vecy)\in L\times L$ such that $(\vecx,\vecy) \in Q_i\times U_i$ holds for some $i$ and $1>\inner{\vecx,\vecy}\ge \cos \theta$ with probability at least 0.9. Alternatively, it needs to find at least $\Omega(n)$ such pairs with overwhelming probability.
    % if \texttt{method}=\texttt{FAS}.
    % \begin{itemize}
    %     \item asked to find, for each $\vecx \in L$, almost all pair $(i,\vecy) \in [t]\times L$ such that $\vecx \in Q_i$, $\vecy \in B_i$, and $\inner{\vecx,\vecy}\ge \cos \theta$ if \texttt{method}=\texttt{Query}, and
    %     \item required to do preprocessing outlined in~\Cref{item:prep1,item:prep2,item:prep3,item:prep4}, and asked to find, for each $i\in [t]$, almost all close pairs $(\vecx,\vecy) \in A_i\times B_i$ if \texttt{method}=\texttt{FAS}.
    % \end{itemize}
\end{definition}
We assume $\theta=\pi/3$ and $t \ge \max(1/\cC_d(\alpha),1/\cC_d(\beta),1/\cW_d(\alpha,\beta,\pi/3))$ to assure that for all near-neighbors have such an index with high probability.

Both \texttt{Query} and \texttt{FAS} methods indeed try to address the above task.
The classical complexity lower bound can be derived from the following lemma. Recall that the classical algorithm always stores the queries and vectors classically.
\begin{theorem}\label{lem: classical lower bound}
    The expected query complexities of the \emph{classical} black-box near-neighbor algorithm with the purpose as in~\Cref{def: method_purpose} is at least
    \begin{align}
        % \Omega\left (
        \underbrace{nt \cdot \cC_d(\beta)}_{\textsl{Preprocess}}+\underbrace{n \cdot \cC_d(\alpha)/\cW_d(\alpha,\beta,\pi/3)}_{\text{sampling filters wrt $Q_i$}} + \underbrace{n^2 \cdot \cC_d(\alpha)\cC_d(\beta)/\cW_d(\alpha,\beta,\pi/3)  }_{\text{inner products}}
        % \right)
    \end{align}
    up to a constant multiplicative factor, or at least $2^{0.2925d+o(d)}$. In particular, it must make $2^{0.2925d+o(d)}$ queries.
    % as follows:
    % \begin{itemize}
    %     \item $\Omega(nt \cC_d(\beta))$ for \textsc{Preprocess} and \Cref{item:prep1,item:prep2,item:prep3,item:prep4} of \textsc{FindAllSolutions},
    %     \item $\Omega(t\cdot \cC_d(\alpha) + nt\cdot  \cC_d(\alpha)\cC_d(\beta))$ for each $\vecx \in L$ for \textsc{Query}, and
    %     \item $\Omega(n^2\cdot \cC_d(\alpha)\cC_d(\beta))$ for each $i\in[t]$ for \textsc{FindAllSolutions}.
    % \end{itemize}
    % The last two items consider the task described in~\Cref{def: method_purpose}.
    % If $\texttt{method}=\texttt{Query}$, 
\end{theorem}
\begin{proof}
    The formal proof should be involved with several probabilistic arguments, 
    and intuitively, the proof shows that the queries in the main body are essential except for the case where the algorithm does not use the hash functions much.
    We give the proof sketch based on the average-case behavior.
    
    The procedures for the first step are fixed, so the first term is obvious. Suppose that the number of inner products made by the algorithm is less than $n^{1.5}\ge 2^{0.2925 d +o(d)}$; otherwise, it collapses to the ``or'' part at the end of the statement.
    
    % Observe that for each near-neighbor $(\vecx,\vecy)$, the number of $i\in[t]$ satisfying the condition of~\Cref{def: method_purpose} is $t \cdot \cW_d(\alpha,\beta,\pi/3)$ in expectation. Also 
    % Observe that the inner product query must ask 
    Note that if the algorithm knows $(\vecx,\vecy)\in Q_i\times U_i$, the probability that $\inner{\vecx,\vecy}\ge 1/2$ is about $1/\cW_d(\alpha,\beta,\pi/3)$ by~\Cref{lem: probs_spherecap}.
    On the other hand, without such a constraint, $\vecx$ and $\vecy$ behave essentially randomly (because we assume that the vectors defining filters are uniformly distributed), the probability that $\inner{\vecx,\vecy}\ge 1/2$ is about $1/(4/3)^{d/2+o(d)}=O(1/n)$ by~\Cref{lem: probs_spherecap}.
    % \footnote{Strictly speaking, if the algorithm knows that $\vecy$ is not included in $U_i$ for $k$ different $i$'s, then the probability becomes at most $\mu(\cC_{\vecx,1/2})/(\mu(\Sd)-k \mu(\cC_{\vecx,1/2}))$. If $k \gg 1/2\cC_d(\beta)$, this could make a difference, but it implies that } 
    Therefore, the number of near-neighbors found by inner product queries without the knowledge of $(\vecx,\vecy)\in Q_i\times U_i$ is at most $n^{1.5} \cdot 1/n=n^{0.5}$ on average. This means that we need to find $n-n^{0.5}=\Omega(n)$ using the inner product queries with \emph{the knowledge of the existence $i$}.

    Fix $\vecx \in L$. 
    For a near-neighbor $\vecy \in L$ of $\vecx$ (which exists with a high probability), the expected number of $i\in[t]$ satisfying the condition of~\Cref{def: method_purpose} is $t \cdot \cW_d(\alpha,\beta,\pi/3)$.
    The algorithm must find one of such $i$ using the sampling filter queries with a probability of at least 0.9. Since the vectors defining filters are random, the algorithm must make $t\cdot \cC_d(\alpha)/(t\cdot \cW_d(\alpha,\beta,\pi/3))$ sampling filter queries. The algorithm also requires making (almost) all inner product queries to find the near-neighbors, which is 
    \begin{align}
        n \cdot\left( \cC_d(\alpha)/( \cW_d(\alpha,\beta,\pi/3)) \right)\cdot \left(n \cdot \cC_d(\beta)\right)
    \end{align}
    which gives the lower bound we desired. 
    Optimizing the complexity is essentially the same as in~\Cref{thm: classical sieving}, which concludes the proof.
    % For the second item, we formally consider the problem of finding, for every $i\in[t]$, $\vecy \in B_i$ such that $\inner{\vecx,\vecy} \ge \theta$ with probability at least 0.9 if exist. Suppose that $\vecy \in L$ is such that $\inner{\vecx,\vecy} \ge \theta$, then there are 
\end{proof}

\subsection{Quantum lower bound}
Now, we turn to the lower bound with a bounded QRAM setting. As outlined before, our strategy is, given a quantum black-box algorithm $ A$, to construct a classical black-box simulation algorithm $B$ that behaves almost like $A$. The result of the simulation can be described as follows.

\begin{theorem}\label{thm: simul}
    Let $A$ be a quantum black-box near-neighbor algorithm following~\Cref{assumption: QRAM} and \Cref{assumption: QRAM_filter}. Suppose that a list of $n=2^{0.2075d+o(d)}$ random vectors $L$ in $\Sd$ is given as input, and that $A$ makes at most $q$ queries to the oracle. Then there exists another black-box near-neighbor algorithm $B$, given the same inputs, which makes at most $2^{2s}\cdot q$ \emph{classical} queries to the oracle such that the output distribution of $B$ and $A$ are identical for any input.
\end{theorem}

Since $B$ always makes classical queries, the lower bound in~\Cref{lem: classical lower bound} is applied to $B$. This implies that the query complexity $q$ of $A$ must satisfy $2^{2s} \cdot q \ge 2^{0.2925d+o(d)}$. Furthermore, if $A$ solves the task given in~\Cref{def: method_purpose}, $B$ must solve the same task since the output distributions are identical. This concludes the proof of~\Cref{thm:qsieve_lowerbound}. We again note that the proof of this theorem is very similar to~\cite[Theorem 4.1]{HYY24}.
\begin{proof}[Proof of~\Cref{thm: simul}]
    Let $A$ be a quantum black-box near-neighbor algorithm following~\Cref{assumption: QRAM} and \Cref{assumption: QRAM_filter}.
    Let $L=(\vecx_1,...,\vecx_n)$. We construct a black-box algorithm $B$ that simulates $A$ as follows.
    \paragraph{Initialization.} For the simulation, the algorithm $B$ maintains a ``labeling function''
    \[
    L:[n] \to [n] \cup \{\bot\}
    \] 
    by abusing the notation; we use $L(i)$ only for the labeling function to avoid confusion.
    Intuitively, $L(i)=\ell$ means, whenever $A$ uses $\vecx_i$ (in the $i$-th entry of $\rX$), the algorithm $B$ uses $\ell$ instead of the vector $\vecx_i$ to construct their quantum states.\footnote{The use of the labeling function is because the algorithm $B$ does not have access to the register $\rX$ and $\rT$ without queries, so it cannot construct the coherent states over the input vectors only using the classical queries; instead, $B$ constructs the coherent states over the labels.}
    In this way, $B$ will always know its exact state at any time of the execution.
    
    All labels are initially set by $\bot$, which means they are not specified yet.
    The algorithm $B$ initializes the function $L$ as follows: For $i=1,...,n$, set $L(\vecx_i)$ uniformly a random element from $[M]$ conditioned that it is not used before. $B$ creates the following state
    \[
    \ket{0,...,0}_{\rW\rQ} \otimes \ket{\vecx_1,...,\vecx_n}_\rX \otimes \ket{\bot,...,\bot}_{\rT'}
    \]
    which is identical to the initial state of $A$ except for the zero vectors in $\rT$ is replaced by $\bot$. Whenever $A$ has the state
    \[
    \sum_{wq,T=(\vect_1,\vect_2,...)} \alpha_{wq,T}\ket{wq}_{\rW\rQ} \otimes \ket{\vecx_1,...,\vecx_n}_\rX \otimes \ket{\vect_1,\vect_2,...}_{\rT}
    \]
    such that $\vect_i = \vecx_{j_i}$ for $i=1,2,..$,
    $B$ will construct the state
    \[
    \sum_{wq,J_T=(j_1,j_2,...)} \alpha_{wq,J_T}\ket{wq}_{\rW\rQ} \otimes \ket{\vecx_1,...,\vecx_n}_\rX \otimes \ket{L(j_1),L(j_2),...}_{\rT'}.
    \]

    \paragraph{Local operation.} When $A$ applies some operation $U$ on its local registers $\rW\rQ$, $B$ also applies the same operation $U$ on its $\rW\rQ.$

    \paragraph{Insertion.} When $A$ applies the insertion query, $B$ does the same query.
    The definition of the insertion query forces to measure the query register, $B$ can apply the same insertion operation made by $A$.

    \paragraph{Sampling vectors from filters.} Suppose $A$ makes the sampling vectors from filters query with the input state
    \[
    \sum_{(C,i,j)\in J,k} \alpha_{C,i,j,k} \ket{C,i,j,k}_{\rQ} \otimes \ket{...,0,...}_\rT
    \]
    for some set $J$.
    For convenience, we assume that $k$ is fixed; the general case can be dealt with analogously. By~\Cref{assumption: QRAM}, $|J|\le 2^s$ holds. 
    
    $B$ makes the classical sampling vectors from filter queries for all $(C,i,j,k)$ for $(C,i,j)\in J$\footnote{$B$ knows the whole state perfectly, so it can be recovered.} and let $\vect_{C,i,j}$ be the sampled vector stored in the $k$-th entry of $\rT$, finds $i_{C,i,j}\in[n]$ such that $\vect_{C,i,j}=\vecx_{i_{C,i,j}}$, and applies the inverse operation of the previous query to remove the $k$-th entry vector from $\rT$. Then, $B$ computes
    \[
    \sum_{(C,i,j)\in J,k} \alpha_{C,i,j,k} \ket{C,i,j,k}_{\rQ} \otimes \ket{...,L(i_{C,i,j}),...}_{\rT'}
    \]
    from the current state $
    \sum_{(C,i,j)\in J,k} \alpha_{C,i,j,k} \ket{C,i,j,k}_{\rQ} \otimes \ket{...,0,...}_{\rT'} $.

    \paragraph{Copying from $\rX$ to $\rT$.} $B$ applies the same operation as $A$.

    \paragraph{Sampling filters from vectors.} This is almost the same as the sampling vectors. This time, the algorithm's state is
    \[
        \sum \alpha_{Z,i,k} \ket{Z,i,j,k}_\rQ \mapsto 
        \sum \alpha_{Z,i,k}\ket{Z,i,\ell_{i},k}_\rQ 
    \]
    for $A$, where $Z\in\{C,Q\}$, $\{Z_{\ell_1},...,Z_{\ell_u}\}$ be the set of relevant filters such that $\ell_1<...<\ell_u$. We omit the register $\rX$ because it is fixed.
    % $B$ must apply
    % \[        
    % \sum \alpha_{Z,i,k}\ket{Z,i,0,k}_\rQ 
    % \mapsto 
    % \sum \alpha_{Z,i,k}\ket{Z,i,\ell_i,k}_\rQ .
    % \]
    Here, the number of $(Z,i,k)$ such that $\alpha_{Z,i,k}\neq0$ must be at most $2^s$ because of~\Cref{assumption: QRAM_filter}. $B$ collects such $(Z,i,k)$ and makes the classical sampling filter queries for all $(Z,i,k)$ and obtain $\ell_i$ (wrt $Z,k$) and construct the corresponding state. It makes at most $2^s$ \emph{classical} sampling filters from vectors queries for each quantum query by $A$.

    \paragraph{Inner product.} When $A$ applies the inner product query on
    \[
        \sum_{r,T} \alpha_{r,T}\ket{r,k,\ell}_\rQ \ket{...,\vect_k,...,\vect_\ell,...}_\rT \mapsto
        \sum_{r,T} \alpha_{r,T}\ket{r\oplus b,k,\ell}_\rQ \ket{...,\vect_k,...,\vect_\ell,...}_\rT,
    \]
    $B$ retrieves the information in the $k$-th and $\ell$-th registers of $\rT'$. Note that the sampling vector queries are the only way to make the coherent state over $\rT$, and by~\Cref{assumption: QRAM}, each entry of $\rT$ has at most $2^s$ vectors with nonzero amplitudes. Thus, we can write $\vect_{k,1},...\vect_{k,2^s},\vect_{\ell,1},...,\vect_{\ell,2^s}$ be such vectors with nonzero amplitudes. $B$ computes the indices $i_{k,1},...,i_{k,2^s},i_{\ell,1},...,i_{\ell,2^s}$ such that $\vecx_{i_{k,\star}} = \vect_{k,\star}$ and the same for $\ell$. Using the classical inner product queries, $B$ compute $1_{\inner{\vecx_{i_{k,a}},\vecx_{i_{\ell,b}}}\ge 1/2} = 1_{\inner{\vect_{{k,a}},\vect_{{k,b}}}\ge 1/2}$ for all $a,b\in [2^s]$. From this, $B$ can compute the same query using $2^{2s}$ classical inner product queries.

    \paragraph{Finalization.} When $A$ outputs a near-neighbor pair $(\vecx_i,\vecx_j)$, it must contain $(i,j)$ in its local register $\rW$. $B$ does the same and outputs $(i,j)$.

    As observed above, the truncated states in register $\rW\rQ$ of $A$ and $B$ are identical, thus the output distributions of $A$ and $B$ are identical for any input.
    The complexity of $B$ is at most $2^{2s}$ times the query complexity of $A$, because each quantum query can be simulated by at most $2^{2s}$ classical query. This concludes the proof.
\end{proof}

\ifnum\draft=1
\newpage
\setcounter{tocdepth}{2}
\tableofcontents{}
\fi

\end{document}